%% file: main.tex
\documentclass[twosided,10pt]{IEEEtran}

\usepackage{authblk}
\usepackage{mathtools}
\usepackage{amsmath}
\usepackage{amssymb}
\usepackage{graphicx}
\usepackage{array}
\usepackage{makecell}
\usepackage{tikz}
\usepackage{pgfplots}	
\usepackage{booktabs}
\usepackage{floatflt}
\usepackage{enumerate}
\usepackage{psfrag}	
\usepackage{array}
\usepackage{multirow,hhline}
\usepackage{exscale}
\usepackage{color}
\usepackage{colortbl}
\usepackage{epsfig,subfigure}
\usepackage{cite}
\usepackage{amsthm}
\usepackage{xfrac}
\usepackage{algorithm}
\usepackage{algpseudocode}
\usepackage[font=small]{caption}
\usepackage[latin1]{inputenc} 
\usepackage[T1]{fontenc} 
\usepackage{enumerate}
\usepackage{siunitx}
\usepackage{nopageno}
\usepackage[left=0.58in, right=0.72in, top=0.7in]{geometry}

\newtheorem{theorem}{Theorem}
\newtheorem{definition}{Definition}
\newtheorem{lemma}{Lemma}

\newtheorem{remark}{Remark}

\usetikzlibrary{shapes,snakes}
\usetikzlibrary{calc}
\usetikzlibrary{patterns}
\usetikzlibrary{decorations.pathmorphing}

\newcommand{\Ali}{\textcolor{black}}

\newcommand{\TxAntenna}[3]{
	\coordinate (a) at (#1,#2);
	\draw[line width=0.25pt,scale=(#3)] (a)--($(a)+(0.2,0)$)--($(a)+(0.2,0.7)$)--
	($(a)+(0.1,0.8)$)--($(a)+(0.3,0.8)$)--($(a)+(0.2,0.7)$);
}

\newcommand{\RxAntenna}[4]{
	\coordinate (a) at (#1,#2);
	\draw[line width=0.25pt,scale=(#3)] (a)--($(a)+(-0.2,0)$)--($(a)+(-0.2,0.7)$)--
	($(a)+(-0.1,0.8)$)--($(a)+(-0.3,0.8)$)--($(a)+(-0.2,0.7)$);
}

\begin{document}

\IEEEoverridecommandlockouts
\title{Robust Transceiver Design for Full-Duplex Decode-and-Forward Relay-Assisted MIMO Systems}
\author{
\IEEEauthorblockN{Hossein Esmaeili*, Ali Kariminezhad*, and Aydin Sezgin}\\
\thanks{
*Equal contribution authors.\\
A. Kariminezhad H. Esmaeili and A. Sezgin are with the Institute of Digital Communication Systems, Ruhr-Universit\"at Bochum (RUB), Germany (emails: \{ali.kariminezhad, hossein.esmaeili aydin.sezgin\}@rub.de).
}}

\maketitle
\thispagestyle{empty}
\begin{abstract}
Robust transceiver design against unresolvable system uncertainties is of crucial importance for reliable communication. For instance, full-duplex communication suffers from such uncertainties when canceling the self-interference, since some residual self-interference (RSI) remains uncanceled due to imperfect channel knowledge. We consider a MIMO multi-hop system, where the source, the relay and the destination are equipped with multiple antennas. The considered decode-and-forward (DF) hybrid relay can operate in either half-duplex or full-duplex mode, and the mode changes adaptively depending on the RSI strength. We investigate a robust transceiver design problem, which maximizes the throughput rate of the worst-case RSI under the self-interference channel uncertainty bound constraint. The yielded problem turns out to be a non-convex optimization problem, where the non-convex objective is optimized over the cone of semidefinite matrices. Without loss of generality, we simplify the problem to the optimization over multiple scalar parameters using majorization theory. Furthermore, we propose an efficient algorithm to obtain a local optimal solution iteratively. Eventually, we obtain insights on the optimal antenna allocation at the relay input-frontend and output-frontend, for relay reception and transmission, respectively. Interestingly, given a number of antennas at the relay, the robustness improves if more antennas are allocated to reception than to transmission.
\end{abstract}
\section{Introduction}
 Reliability and throughput are two of the most crucial requirements for the next generation of wireless networks.  Optimally relaying  the  signal  from  a  source  to  a  destination  can  help enhance  reliability  and  capacity  of  networks  and  is  currently an  active  research  area~\cite{Kariminezhad2017}. Furthermore, relaying is the only communication means in disaster scenarios if the direct source-destination link is not available. Exploiting a relay for improving communication throughput rate raises several questions to be answered. For instance, how should the relay process the received signal before dispatching it to the destination? Now, relay can receive a signal from the source, process it and transmit it towards the destination in a successive manner. This type of relaying technique is known as half-duplex relaying. Alternatively, while receiving a signal at a certain time instant, a relay can simultaneously transmit the previously received signals. This technique is known as full-duplex relaying~\cite{Bliss2007}. Authors in~\cite{zhang2016full} discuss a range of potential FD techniques. In ~\cite{della2019distributed}, authors propose a scheme for users in Device-to-Device enabled in-band FD networks. 

As a consequence of transmitting and receiving at a common resource unit, the relay is confronted with self-interference (SI). Note that, full-duplex relaying potentially increases the total throughput rate of the communication compared to the half-duplex counterpart, only if the SI is handled properly at the relay input. By physically isolating the transmitter and receiver frontends of the relay, a significant portion of SI can be reduced~\cite{Sabharwal2014,Shankar2012}. Moreover, analog and/or digital signal processing at the relay input can be utilized to cancel a portion of SI~\cite{Shankar20122, Bliss2012, Eltawil2015,Vogt2018,Lee2014,Irigaray2018}. This can be realized if the estimate of the SI channel state information (CSI) can be obtained at the relay. \Ali{These SI cancellation procedures can effectively mitigate the destructive impact of SI up to a certain level. Hence, the remaining portion, so-called residual self-interference (RSI), is still observed at the relay input. The distribution of the RSI is investigated in~\cite{Irio2018,Irio2019}. The authors in~\cite{Alexandris2014} study the impact of RSI on practical setup. Moreover, the authors in ~\cite{Masmoudi2014} categorize the RSI sources in full-duplex in-band communication. This RSI is mainly due to the channel estimation uncertainties and also the transmitter noise. Therefore, the quality of channel estimation plays an important role for limiting RSI if the conventional modulation techniques are utilized. Interestingly, the authors in~\cite{Koohian2017} employ a superimposed signaling procedure (asymmetric modulation constellation) in the basic point-to-point FD communication for cancelling the SI and further retrieving the desired information contents without requiring channel estimates. The RSI degrades the performance of the communication quality evidently. The authors in~\cite{Chae2017} study the degrees-of-freedom (DoF) performance of FD cellular network in the presence of RSI. Furthermore, the authors in~\cite{Herath2013} with such degradation. Having the RSI, the authors in~\cite{Zlatanov2017} study the capacity of Gaussian two-hop FD relay.} 

By exploiting multiple antennas at the relay, the throughput rate from the source to destination can be improved~\cite{Fan2007,Mo2012}. Using multiple antennas at the relay provides the feasibility of SI cancellation spatially by beamforming techniques such that the impact of SI can be mitigated~\cite{Riihonen2011,Lioliou2010}. For instance, zero-forcing (ZF) beamforming forces the SI to zero at the relay input, however, it is not an optimal scheme in weak SI regimes if the relay is equipped with a limited number of antennas. Here, they show the optimality of ZF process at the relay with a very large antenna array. In contrast, ZF process at the relay is shown to be optimal when there are a massive number of antennas at the relay~\cite{Ngo2014,Kariminezhad2017SCC}.

Further, exploiting multiple antennas at the source and destination can provide the opportunity for improving the communication throughput rate. In a MIMO multi-hop system, the authors in~\cite{Suraweera2014} investigate a amplify-and-forward (AF) relay, where the precoder at the relay and the decoder at the destination are jointly optimized for maximizing the source-destination throughput rate. \Ali{Moreover, the authors in~\cite{Chen2016} study the power allocation problem in two-hop decode-and-forward (DF) MIMO FD relaying.} These works mainly assume a single stream transmission, which is not always optimal. The authors in~\cite{Jeong2017} consider a MIMO decode-and-forward (DF) relaying scheme with energy harvesting demands at the relay fulfilled by the source. \Ali{Assuming FD multi-pair communication with  multi-antenna transceivers, the authors in~\cite{Cirik2015} study the weighted sum-rate maximization, where they rendered the problem to the weighted mean squared error (MSE) minimization for obtaining low-complex algorithm.} These works mainly assume the availability of the SI channel for optimal MIMO pre- and post-processing tasks, where the RSI is simply treated as noise with estimated statistical moments. However, these estimates can not be guaranteed to be valid for all applications and scenarios. Hence, the study of a robust design becomes crucially important.

Robust transceiver design against the worst-case RSI channel helps find the threshold for switching between HD and FD operating modes in hybrid relay systems. The authors in~\cite{Taghizadeh2014} investigate a robust design for multi-user full-duplex relaying with multi-antenna DF relay. In that work, the sources and destinations are equipped with single antennas. Moreover, the authors in~\cite{Cirik2016} investigate a robust transceiver design for FD multi-user MIMO systems for maximizing the weighted sum-rate of the network.

\textit{Contribution:} We consider a DF multi-hub system with multiple antennas at the source, relay and destination. In this system, we allow multi-stream beamforming for throughput rate maximization. The optimization of maximum achievable rate of the DF full-duplex relaying is cast as a non-convex optimization problem. \Ali{The complexity of this problem is shown to be reduced analytically using majorization theory.} We propose an efficient algorithm to solve this problem in polynomial time. Finally, the transmit signal covariances at the source and the relay are designed efficiently to improve robustness against worst-case RSI channel in a given uncertainty bound.

\section{System Model}
We consider the communication from a source equipped with $M$ antennas to a destination with $N$ antennas. The reliable communication is assumed to be only feasible by means of a relay with $K_\mathrm{t}$ transmitter and $K_{r}$ receiver antennas at the output and input frontends, respectively. The received signals at the relay and destination are given by
\begin{align}
\mathbf{y}_{r}&= \mathbf{H}_1\mathbf{x}_{s}+\kappa \mathbf{H}_{r}\mathbf{x}_{r}+\mathbf{n}_{r},\\
\mathbf{y}_\mathrm{d}&= \mathbf{H}_2\mathbf{x}_{r}+\mathbf{n}_\mathrm{d},
\end{align}
respectively, where $\kappa\in\{0,1\}$. Notice that, $\kappa=0$ coincides with HD relaying and $\kappa=1$ denotes FD relaying. The transmit signal of the source is denoted by $\mathbf{x}_{s}\in\mathbb{C}^{M}$ with the covariance matrix $\mathbf{Q}_{s}=\mathbb{E}[\mathbf{x}_{s}\mathbf{x}^H_{s}]$, and the transmit signal of the relay is represented by $\mathbf{x}_{r}\in\mathbb{C}^{K_{\mathrm{t}}}$, with the covariance matrix $\mathbf{Q}_{r}=\mathbb{E}[\mathbf{x}_{r}\mathbf{x}^H_{r}]$. The additive noise vectors at the relay and destination are denoted by $\mathbf{n}_{r}\in\mathbb{C}^{K_{r}}$ and $\mathbf{n}_\mathrm{d}\in\mathbb{C}^N$, respectively, which are assumed to follow zero-mean Gaussian distributions with identity covariance matrices. The source-relay channel is represented by $\mathbf{H}_1\in\mathbb{C}^{K_\mathrm{t}\times M}$ and the relay-destination channel is denoted by $\mathbf{H}_2\in\mathbb{C}^{N\times K_{r}}$, see~\figurename{ \ref{fig:SystemModel}}. These channels are assumed to be perfectly known. Furthermore, the self-interference (SI) channel at the relay is represented by $\mathbf{H}_{r}$, which is assumed to be known only imperfectly. In what follows, we present the achievable throughput rates for the HD and FD relaying. In the next section, we start with the HD relay, in which $\kappa=0$.
\begin{figure}
\centering
\tikzset{every picture/.style={scale=.95}, every node/.style={scale=0.7}}%
\input{SystemModel}
\vspace*{0.4cm}
\caption{System model of a full-duplex relay}
\label{fig:SystemModel}
\end{figure}
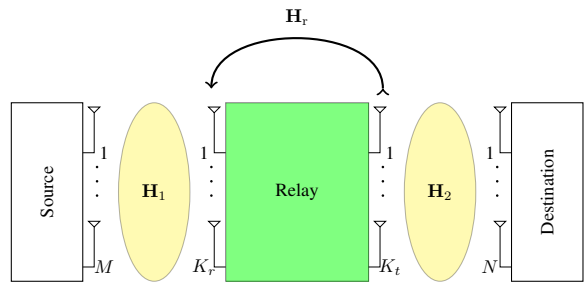

\section{Achievable Rate (Half-Duplex Relay)}
Suppose that the relay employs DF strategy. We consider a simple half-duplex relay, where the source and the relay transmit in two subsequent time instances. Using time sharing, the achievable rate between the source and destination nodes is given by
\begin{align}
R^{\mathrm{HD}}=\min(\alpha R^{\mathrm{HD}}_\mathrm{sr},(1-\alpha)R^{\mathrm{HD}}_\mathrm{rd}),
\end{align}
in which $R^{\mathrm{HD}}_\mathrm{sr}$ and $R^{\mathrm{HD}}_\mathrm{rd}$ are the achievable rates on the source-relay and relay-destination links, respectively, and $\alpha$ is the time-sharing parameter. Note that, in half-duplex relaying the source and relay transmissions are conducted in separate channel uses. Thus, these rates are given by
\begin{align}
R^{\mathrm{HD}}_\mathrm{sr}&=\log\big|\mathbf{I}_{K_{r}}+\mathbf{H}_1\mathbf{Q}_{s}\mathbf{H}^H_1\big|,\\
R^{\mathrm{HD}}_\mathrm{rd}&=\log\big|\mathbf{I}_{N}+\mathbf{H}_2\mathbf{Q}_{r}\mathbf{H}^H_2\big|.
\end{align}
Now, the transmit covariance matrices $\mathbf{Q}_{s}\in\mathbb{H}^{M\times M}$ and $\mathbf{Q}_{r}\in\mathbb{H}^{K_\mathrm{t}\times K_\mathrm{t}}$ are optimized by maximizing the achievable rate from the source to the destination. Here, the convex cone of Hermitian positive semidefinite matrices of dimensions $M\times M$ and $K_\mathrm{t}\times K_\mathrm{t}$ are represented by $\mathbb{H}^{M\times M}$ and $\mathbb{H}^{K_\mathrm{t}\times K_\mathrm{t}}$, respectively. Importantly, for maximizing this achievable rate, the time-sharing parameter, i.e., $\alpha$ needs to be optimized alongside the system parameters, e.g., power allocation. Readily, optimal $\alpha$ occurs at $\alpha R^{\mathrm{HD}}_\mathrm{sr}=(1-\alpha)R^{\mathrm{HD}}_\mathrm{rd}$. Therefore, the achievable rate becomes as follows,
\begin{align}
R^{\mathrm{HD}}=\frac{R^{\mathrm{HD}}_\mathrm{sr}R^{\mathrm{HD}}_\mathrm{rd}}{R^{\mathrm{HD}}_\mathrm{sr}+R^{\mathrm{HD}}_\mathrm{rd}}.
\end{align}
The throughput rate maximization problem is cast as
\begin{subequations}\label{P:HDa}
\begin{align}
\max_{\mathbf{Q}_{s},\mathbf{Q}_{r}}\quad & \frac{R^{\mathrm{HD}}_\mathrm{sr}R^{\mathrm{HD}}_\mathrm{rd}}{R^{\mathrm{HD}}_\mathrm{sr}+R^{\mathrm{HD}}_\mathrm{rd}}\\
\text{subject to}\quad\quad & \mathrm{Tr}(\mathbf{Q}_{s})\leq P_{s},\label{P:HDa:ConsA}\\ 
&\mathrm{Tr}(\mathbf{Q}_{r})\leq P_{r},\label{P:HDa:ConsB}
\end{align}
\end{subequations}
in which the constraints~\eqref{P:HDa:ConsA} and~\eqref{P:HDa:ConsB} represent the transmit power constraints and $P_{s}$ and $P_{r}$ are the transmit power budgets at the source and relay, respectively. Let $\mathbf{Q}_{s}=\mathbf{U}_{s}\boldsymbol\Gamma_{s}\mathbf{U}^H_{s}$ and $\mathbf{Q}_{r}=\mathbf{U}_{r}\boldsymbol\Gamma_{r}\mathbf{U}^H_{r}$. Since, $R^{\mathrm{HD}}_\mathrm{sr}$ and $R^{\mathrm{HD}}_\mathrm{rd}$ are concave functions of $\mathbf{Q}_{s}$ and $\mathbf{Q}_{r}$, the solutions are given as~\cite{Telatar99}
\begin{align}
\mathbf{Q}^{\star}_{s}={\bf U}^{\star}_{s}{\bf\Gamma}^{\star}_{s}{\bf U}^{{\star}^H}_{s},\ \text{with}\ {\bf U}^\star_{s}=\mathbf{R}_1,\label{eq:QsA}\\
\mathbf{Q}^{\star}_{r}={\bf U}^{\star}_{r}{\bf\Gamma}^{\star}_{r}{\bf U}^{{\star}^H}_{r},\ \text{with}\ {\bf U}^\star_{r}=\mathbf{R}_2.\label{eq:QrA}
\end{align}
Notice that ${\bf R}_1$ and ${\bf R}_2$ correspond to the right singular matrices of ${\bf H}_1$ and ${\bf H}_2$, respectively, with
$
{\bf H}_1={\bf L}_1{\bf \Sigma}_1{\bf R}^H_1,$ and $
{\bf H}_2={\bf L}_2{\bf \Sigma}_2{\bf R}^H_2.
$
The diagonal matrices ${\bf\Gamma}^{\star}_{s}$ and ${\bf\Gamma}^{\star}_{r}$ are determined by the water-filling algorithm~\cite{Telatar99} as
\begin{align}
{\bf\Gamma}^{\star}_{s}&=\left(\tau_{s}\mathbf{I}-({\bf\Sigma}^{H}_1{\bf\Sigma}_1)^{-1} \right)^{+},\\
{\bf\Gamma}^{\star}_{r}&=\left(\tau_{r}\mathbf{I}-({\bf\Sigma}^{H}_2{\bf\Sigma}_2)^{-1} \right)^{+},
\end{align}
respectively. The water levels $\tau_{s}$ and $\tau_{r}$ are chosen such that they satisfy the power constraint, i.e., $\mathrm{Tr}\left(\tau_{s}\mathbf{I}-({\bf\Sigma}_1{\bf\Sigma}^{H}_1)^{-1} \right)=P_{s}$, and $\mathrm{Tr}\left(\tau_{r}\mathbf{I}-({\bf\Sigma}_2{\bf\Sigma}^{H}_2)^{-1} \right)=P_{r}$. Next, we determine the maximum achievable rate for the full-duplex relay.  

\section{Achievable Rate (Full-Duplex Relay)}
In this case both links are active at the same time. As a result, the signals from the relay transmitter interfere with the receiving signal at the relay receiver. We assume that an estimate of the self-interference (SI) channel $\mathbf{H}_{r}$ is available at the relay denoted by $\hat{\mathbf{H}}_{r}$. Hence, the unknown channel estimation error (residual self-interference channel) represented by $\bar{\mathbf{H}}_{r}$ is given as
\begin{align}
\bar{\mathbf{H}}_{r}=\mathbf{H}_{r}-\hat{\mathbf{H}}_{r}.
\end{align}
In this work, we assume that some portion of the SI is canceled based on the available estimate $\hat{\mathbf{H}}_{r}$, such that only a residual self-interference (RSI) remains. Here, we represent this portion by $\bar{\mathbf{H}}_{r}\mathbf{x}_{r}$. Considering a full-duplex decode-and-forward relay, the following rate is achievable
\begin{align}
R^{\mathrm{FD}}=\min(R^{\mathrm{FD}}_\mathrm{sr},R^{\mathrm{FD}}_\mathrm{rd}),
\end{align}
in which
\begin{align}
R^{\mathrm{FD}}_\mathrm{sr}&=\log_2\frac{\big|\mathbf{I}_{K_{r}}+\mathbf{H}_1\mathbf{Q}_{s}\mathbf{H}^H_1+\bar{\mathbf{H}}_{r}\mathbf{Q}_{r}\bar{\mathbf{H}}^H_{r}\big|}{\big|\mathbf{I}_{K_{r}}+\bar{\mathbf{H}}_{r}\mathbf{Q}_{r}\bar{\mathbf{H}}^H_{r}\big|},\label{eq:FD_srA}\\
R^{\mathrm{FD}}_\mathrm{rd}&=\log_2\big|\mathbf{I}_{N}+\mathbf{H}_2\mathbf{Q}_{r}\mathbf{H}^H_2\big|.
\end{align}
Notice that, with perfect SI channel state information, the SI could be completely removed from the received signal at the relay input-frontend. However, assuming that the RSI remains uncanceled, a robust transceiver against the worst-case RSI channel is required which is formulated as an optimization problem as follows
\begin{subequations}\label{P:FDa}
\begin{align}
 \max_{\mathbf{Q}_{s},\mathbf{Q}_{r}}\  \min_{\bar{\mathbf{H}}_{r}}\quad & \min\bigg( R^{\mathrm{FD}}_\mathrm{sr}, R^{\mathrm{FD}}_\mathrm{rd} \bigg) \tag{\ref{P:FDa}}\\
\text{subject to}\quad\quad & \mathrm{Tr}(\mathbf{Q}_{s})\leq P_{s},\label{P:FDa:ConsA}\\ 
&\mathrm{Tr}(\mathbf{Q}_{r})\leq P_{r},\label{P:FDa:ConsB}\\
&\mathrm{Tr}(\bar{\mathbf{H}}_{r}\bar{\mathbf{H}}^H_{r})\leq T,\label{P:FDa:ConsC}
\end{align}
\end{subequations}

in which the throughput rate with respect to the worst-case RSI channel is maximized. In constraint~\eqref{P:FDa:ConsC}, $T$ represents the RSI channel uncertainty bound. Notice that, $\mathrm{Tr}(\bar{\mathbf{H}}_{r}\bar{\mathbf{H}}^H_{r})$ represents the sum of the squared singular values of $\mathbf{H}_\mathrm{r}$. It should be noted that, using a bounded matrix norm is the most common way for modeling the uncertainty of a matrix~\cite{wang2013robust,shen2013robust}. Next, we investigate the optimal design for the full-duplex relay with the worst-case RSI.

Using the following theorem and lemma, we show that for every possible choice of ${\mathbf{H}}_\mathrm{1}$ and ${\mathbf{H}}_\mathrm{2}$, there exists at least one set of simultaneously diagonalizable matrices ${\mathbf{H}}_{r}$, ${\mathbf{Q}}_{s}$ and ${\mathbf{Q}}_{r}$ that are the solutions to the problem~\eqref{P:FDa}.
%
%
\begin{lemma}\label{lem:1}
For two positive semi-definite and positive definite matrices $\mathbf{A}$ and $\mathbf{B}$ with eigenvalues $\lambda_1\left(\mathbf{A}\right)\geq \lambda_2\left(\mathbf{A}\right) \geq ... \geq \lambda_N\left(\mathbf{A}\right) $ and $\lambda_1 \left(\mathbf{B}\right) \geq \lambda_2\left(\mathbf{B}\right) \geq ... \geq \lambda_N\left(\mathbf{B}\right)$ respectively, the following inequalities hold,
\begin{equation}
    \prod_{i = 1}^{N} \left( 1 + \frac{\lambda_i\left(\mathbf{A}\right)}{\lambda_i\left(\mathbf{B}\right)} \right) \leq \Big| \mathbf{I}+\mathbf{A}\mathbf{B}^{-1} \Big| \leq \prod_{i = 1}^{N} \left( 1 + \frac{\lambda_i\left(\mathbf{A}\right)}{\lambda_{N+1-i}\left(\mathbf{B}\right)} \right).\label{eq:Fiedler0}
\end{equation}

\begin{proof}
Consider the Fiedler's inequality given by~\cite{fiedler1971bounds},
\begin{equation}
\prod_{i = 1}^{N} \left( \lambda_i\left(\mathbf{A}\right) + \lambda_i\left(\mathbf{B}\right) \right) \leq \Big| \mathbf{B}+\mathbf{A} \Big| \leq \prod_{i = 1}^{N} \left( \lambda_i\left(\mathbf{A}\right) + \lambda_{N+1-i}\left(\mathbf{B}\right) \right).\label{eq:Fiedler}
\end{equation}
Furthermore, given $\mathbf{B}$ as a positive definite matrix, the followings hold,
\begin{align}
\big| \mathbf{B}^{-1} \big| &> 0,\\
\big| \mathbf{B}^{-1} \big| &= \prod_{i = 1}^{N}\frac{1}{\lambda_i\left(\mathbf{B}\right)}.
\end{align}
Now, dividing the sides of~\eqref{eq:Fiedler} by $\big|\mathbf{B}\big|$, one can readily obtain~\eqref{eq:Fiedler0}.
\end{proof}
\end{lemma}

Note that, in~\eqref{eq:Fiedler0} the inequalities hold with equalities if and only if $\mathbf{A}$ and $\mathbf{B}$ are diagonalizable over a common basis. Using the result of lemma~\ref{lem:1}, one can obtain
\begin{align}
       \log_2 \big|\mathbf{I}_{K_{r}}+&\mathbf{H}_1\mathbf{Q}_{s}\mathbf{H}^H_1   \left(\mathbf{I}+\bar{\mathbf{H}}_{r}\mathbf{Q}_{r}\bar{\mathbf{H}}^H_{r}\right)^{-1}\big|
        \nonumber\\ &\geq \sum_{i = 1}^{\min{(M,K_r)}} \log_2 \left( 1+ \frac{\lambda_i\left(\mathbf{H}_1\mathbf{Q}_{s}\mathbf{H}^H_1\right)}{\lambda_i\left(\mathbf{I}+\bar{\mathbf{H}}_{r} \mathbf{Q}_{r} \bar{\mathbf{H}}^H_{r}\right)} \right). \label{eq:InequalityA}
\end{align}
Also it holds that $\lambda_i\left(\mathbf{I}+\bar{\mathbf{H}}_{r}\mathbf{Q}_{r}\bar{\mathbf{H}}^H_{r}\right) = 1 + \lambda_i\left(\bar{\mathbf{H}}_{r}\mathbf{Q}_{r}\bar{\mathbf{H}}^H_{r}\right)$. Hence, we obtain
\begin{align}
       \log_2\big|\mathbf{I}_{K_{r}}+&\mathbf{H}_1\mathbf{Q}_{s}\mathbf{H}^H_1   \left(\mathbf{I}+\bar{\mathbf{H}}_\mathrm{r}\mathbf{Q}_{r}\bar{\mathbf{H}}^H_\mathrm{r}\right)^{-1}\big|
        \nonumber\\ &\geq \sum_{i = 1}^{\min{(M,K_r)}} \log_2 \left( 1+ \frac{\lambda_i\left(\mathbf{H}_1\mathbf{Q}_{s}\mathbf{H}^H_1\right)}{1+\lambda_i\left(\bar{\mathbf{H}}_{r}\mathbf{Q}_{r}\bar{\mathbf{H}}^H_{r}\right)} \right). \label{eq:InequalityAA}
\end{align}
Note that, the inequality holds with equality whenever $\mathbf{H}_1\mathbf{Q}_\mathrm{s}\mathbf{H}^H_1$ and $\left(\mathbf{I}+\bar{\mathbf{H}}_\mathrm{r}\mathbf{Q}_\mathrm{r}\bar{\mathbf{H}}^H_\mathrm{r} \right)^{-1}$ share a common basis which in turn, means $\mathbf{H}_1\mathbf{Q}_\mathrm{s}\mathbf{H}^H_1$ and $\bar{\mathbf{H}}_\mathrm{r}\mathbf{Q}_\mathrm{r}\bar{\mathbf{H}}^H_\mathrm{r}$ occupy a common basis. Now, instead of doing the minimization over the left-hand side (LHS) of equation~\eqref{eq:InequalityAA}, we can first minimize the right-hand side (RHS) to find the optimum eigenvalues. Then show that there exists matrices with those optimum eigenvalues and as a result, the inequality becomes equality.
\begin{remark}\label{Remark1}
Having equality $\mathbf{C}=\bar{\mathbf{H}}_{r}\mathbf{Q}_{r}\bar{\mathbf{H}}^H_{r}$, one can generally conclude $\det \left(\mathbf{C}\right)=\det\left(\bar{\mathbf{H}}^H_{r}\bar{\mathbf{H}}_{r}\right)\det\left(\mathbf{Q}_{r}\right)$. Further, using the properties of determinant we can also conclude $\prod_{i=1}^N\lambda_i\left(\mathbf{C}\right)=\prod_{i=1}^N\left(\lambda_{\rho\left(i\right)}\left(\bar{\mathbf{H}}^H_{r}\bar{\mathbf{H}}_{r}\right)\lambda_i\left(\mathbf{Q}_{r}\right)\right)$ where ${\rho\left(i\right)}$ is a random permutation of $i$ and indicates that there is no need for $\lambda_{\rho\left(i\right)}\left(\bar{\mathbf{H}}^H_{r}\bar{\mathbf{H}}_{r}\right)$ to be in decreasing order. However, one cannot generally conclude $\lambda_i\left(\mathbf{C}\right)=\lambda_{\rho\left(i\right)}\left(\bar{\mathbf{H}}^H_{r}\bar{\mathbf{H}}_{r}\right)\lambda_i\left(\mathbf{Q}_{r}\right)$ for every single $i$, unless $\bar{\mathbf{H}}^H_{r}\bar{\mathbf{H}}_{r}$ and $\mathbf{Q}_{r}$ share common basis.
\end{remark}
As a result of Remark~\ref{Remark1}, in a general case, we cannot rewrite~\eqref{eq:InequalityAA} in terms of $\lambda_i\left(\bar{\mathbf{H}}^H_{r}\bar{\mathbf{H}}_{r}\right)$, $\lambda_i\left(\mathbf{Q}_{r}\right)$, $\lambda_i\left({\mathbf{H}}^H_\mathrm{1}{\mathbf{H}}_\mathrm{1}\right)$ and $\lambda_i\left(\mathbf{Q}_{s}\right)$. However, if we show that for every choice of $\mathbf{Q}_{s}$, there exists a matrix $\mathbf{Q}'_{s}$ with properties: 1) $\lambda_i({\mathbf{H}}_\mathrm{1}\mathbf{Q}_{s}{\mathbf{H}}^H_\mathrm{1})=\lambda_i({\mathbf{H}}_\mathrm{1}\mathbf{Q}'_{s}{\mathbf{H}}^H_\mathrm{1})$; 2) $\lambda_i({\mathbf{H}}_\mathrm{1}\mathbf{Q}'_{s}{\mathbf{H}}^H_\mathrm{1}) = \lambda_i(\mathbf{Q}'_{s})\lambda_i({\mathbf{H}}^H_\mathrm{1}{\mathbf{H}}_\mathrm{1})$ and 3) $\text{Tr}(\mathbf{Q}'_{s}) \leq \text{Tr}(\mathbf{Q}_{s})$; then we can use $\mathbf{Q}'_{s}$ instead and rewrite ~\eqref{eq:InequalityAA} in terms of $\lambda_i\left({\mathbf{H}}^H_\mathrm{1}{\mathbf{H}}_\mathrm{1}\right)$ and $\lambda_i\left(\mathbf{Q}'_{s}\right)$ to simplify the problem. The first property implies that both $\mathbf{Q}_{s}$ and $\mathbf{Q}'_{s}$ have the exact same impact on the capacity. Hence, if we find  a $\mathbf{Q}_{s}$ which is the solution to the problem~\eqref{P:FDa}, its corresponding $\mathbf{Q}'_{s}$ will also be a solution. The second property means, unlike $\mathbf{Q}_{s}$, $\mathbf{Q}'_{s}$ actually shares the common basis with ${\mathbf{H}}^H_\mathrm{1}{\mathbf{H}}_\mathrm{1}$. The last property implies that $\mathbf{Q}'_{s}$ is at least as good as $\mathbf{Q}_{s}$ in terms of power consumption. Observe that if we show for every feasible $\mathbf{Q}_{s}$ there exists at least one such $\mathbf{Q}'_{s}$, then we can solve the problem~\eqref{P:FDa} in a much easier way. The reason is, in such a case, instead of searching for optimal $\mathbf{Q}_{s}$ on the whole feasible set, we can search for the optimal $\mathbf{Q}'_{s}$. Unlike $\mathbf{Q}_{s}$, finding $\mathbf{Q}'_{s}$ does not need a complete search on the whole feasible set since $\mathbf{Q}'_{s}$ shares the common basis with ${\mathbf{H}}^H_\mathrm{1}{\mathbf{H}}_\mathrm{1}$. Therefore, we can limit our search only on the portion of the feasible set in which matrices have eigendirections identical to those of ${\mathbf{H}}^H_\mathrm{1}{\mathbf{H}}_\mathrm{1}$. Similarly, if we show for every choice of $\bar{\mathbf{H}}_{r}$, there exist at least one $\bar{\mathbf{H}}'_{r}$ for which we have three conditions $\lambda_i(\bar{\mathbf{H}}_{r}\mathbf{Q}_{r}\bar{\mathbf{H}}^H_{r})=\lambda_i(\bar{\mathbf{H}}'_{r}\mathbf{Q}_{r}\bar{\mathbf{H}}'^H_{r})$, $\lambda_i(\bar{\mathbf{H}}_{r}\mathbf{Q}_{r}\bar{\mathbf{H}}'^H_{r}) = \lambda_i(\mathbf{Q}_{r})\lambda_i(\bar{\mathbf{H}}'^H_\mathrm{1}\bar{\mathbf{H}}'_\mathrm{1})$ and $\text{Tr}(\bar{\mathbf{H}}'^H_\mathrm{1}\bar{\mathbf{H}}'_\mathrm{1}) \leq \text{Tr}(\bar{\mathbf{H}}^H_\mathrm{1}\bar{\mathbf{H}}_\mathrm{1})$, we can simplify our search to finding $\bar{\mathbf{H}}'_{r}$ instead of $\bar{\mathbf{H}}_{r}$. In the next theorem, we show such $\mathbf{Q}'_{s}$ and $\bar{\mathbf{H}}'_{r}$ exist.


\begin{theorem}\label{Theorem:1}
For every matrices $\mathbf{Q}_{s}$ and ${\mathbf{H}}_\mathrm{1}$, there exists at least one matrix $\mathbf{Q}'_{s}$ that satisfies the following conditions, 
\begin{align}
    \lambda_i \big( \mathbf{H}_1\mathbf{Q}_{s}\mathbf{H}^H_1 \big) &= \lambda_i \big( \mathbf{H}_1\mathbf{Q}'_{s}\mathbf{H}^H_1 \big),\label{P:newo1a}\\ 
    \lambda_i({\mathbf{H}}_\mathrm{1}\mathbf{Q}'_{s}{\mathbf{H}}^H_\mathrm{1}) &= \lambda_{\rho\left(i\right)}(\mathbf{Q}'_{s})\lambda_i({\mathbf{H}}^H_\mathrm{1}{\mathbf{H}}_\mathrm{1}),\label{P:newo1b}\\
    \text{Tr}(\mathbf{Q}'_{s}) &\leq \text{Tr}(\mathbf{Q}_{s}).\label{P:newo1c}
\end{align} \label{P:newo1}
Similarly, for every $\bar{\mathbf{H}}_{r}$ and $\mathbf{Q}_{r}$, there exists at least one matrix $\bar{\mathbf{H}}'_{r}$  for which we have
\begin{align}   
    \lambda_i \big( \bar{\mathbf{H}}_{r}\mathbf{Q}_{r}\bar{\mathbf{H}}^H_{r} \big) &= \lambda_i \big( \bar{\mathbf{H}}'_{r}\mathbf{Q}_{r}\bar{\mathbf{H}}'^H_{r} \big),\\
    \lambda_i\left(\bar{\mathbf{H}}'_{r}\mathbf{Q}_{r}\bar{\mathbf{H}}'^H_{r}\right)&=\lambda_{\rho\left(i\right)}\left(\bar{\mathbf{H}}'_{r}\bar{\mathbf{H}}'^H_{r}\right)\lambda_i\left(\mathbf{Q}_{r}\right),\\
    \text{Tr}(\bar{\mathbf{H}}'_{r}\bar{\mathbf{H}}'^H_{r}) &\leq \text{Tr}(\bar{\mathbf{H}}_{r}\bar{\mathbf{H}}^H_{r}).\label{P:newo2}
\end{align}
\begin{proof}
The proof is given in Appendix I.
\end{proof}
\end{theorem}
For the sake of simplicity, we use following notions for the rest of the paper,
\begin{align}
&\lambda_i(\mathbf{Q}_{s}) = \gamma_{s_i},\\
&\lambda_i(\mathbf{Q}_{r}) = \gamma_{r_i},\\
&\lambda_i({\mathbf{H}}_\mathrm{1}{\mathbf{H}}'^H_\mathrm{1}) = \sigma^2_{1_i},\\
&\lambda_i(\bar{\mathbf{H}}_{r}\bar{\mathbf{H}}'^H_{r}) = \sigma^2_{r_i}.
\end{align}

Now, using Theorem~\ref{Theorem:1} alongside Lemma~\ref{lem:1}, we infer that with no loss of generality, we can use $\mathbf{Q}'_{s}$ and $\bar{\mathbf{H}}'_{r}$ instead of $\mathbf{Q}_{s}$ and $\bar{\mathbf{H}}_{r}$, respectively. This is helpful because all $\mathbf{Q}'_{s}$, ${\mathbf{H}}^H_\mathrm{1}{\mathbf{H}}_\mathrm{1}$, $\mathbf{Q}_{r}$ and $\bar{\mathbf{H}}'^H_{r}\bar{\mathbf{H}}'_{r}$ share the common basis. Hence, problem \eqref{P:FDa} is reformulated as,
\begin{subequations}\label{P:FDbbb}
\begin{align}
\max_{\boldsymbol{\gamma_{s}},\boldsymbol{\gamma_{r}}} \quad  \min_{\boldsymbol{\sigma_{r}}}\quad & \min\Bigg( \sum_{i=1}^{\min(M,K_{r})} \!\log{\!\left(1+
\frac{\sigma_{1_i}^2 \gamma_{{s}_{{\rho{\left(i\right)}}}}}{1+\gamma_{{r}_i}\sigma_{{r}_{{\rho{\left(i\right)}}}}^2}\right)},\nonumber\\ &{\hspace*{1cm}\sum_{i=1}^{{\min(K_{\mathrm{t}},N)}} \!\log{\!\Big(1+\sigma^2_{2_i} \gamma_{r_i}\Big)} \Bigg)}\\
\text{subject to}\quad & \|\boldsymbol{\gamma}_{s}\|_1\leq P_{s},\label{34A}\\
&\|\boldsymbol{\gamma}_{r}\|_1 \leq P_{r},\label{34B}\\
&\|\boldsymbol{\sigma}^2_{r}\|_1\leq T,\label{34C}\\
&\sigma_{1_i}^2\gamma_{{s}_{\rho{\left(i\right)}}} \geq \sigma_{1_{i+1}}^2\gamma_{{s}_{\rho{\left(i+1\right)}}},\ \forall i \leq \min(M,K_\mathrm{r}),\label{P:FDb:ConsA}\\
&\gamma_{\mathrm{r}_i}\sigma_{\mathrm{r}_{\rho{\left(i\right)}}}^2 \geq \gamma_{\mathrm{r}_{i+1}}\sigma_{\mathrm{r}_{\rho{\left(i+1\right)}}}^2,\ {\forall i \leq \min(K_{\mathrm{t}},N).}\label{P:FDb.ConsB}
\end{align}
\end{subequations}
Note that, the two additional constraints~\eqref{P:FDb:ConsA} and~\eqref{P:FDb.ConsB} need to be satisfied due to the conditions of Lemma~\ref{lem:1} (i.e. eigenvalues have to be in decreasing order). 
Interestingly, these two additional constraints are affine. The above optimization problem can further be simplified using the following lemma,
\begin{lemma}\label{lemm2}
The objective function of the optimization problem~\eqref{P:FDbbb} is optimized when the constraints\eqref{34A} and~\eqref{34C} are satisfied with equality.
\end{lemma}
\begin{proof}
Proof is given in the Appendix II.
\end{proof}
Exploiting lemma~\ref{lemm2}, problem~\eqref{P:FDbbb} is now reduced to,
\begin{subequations}\label{P:FDbA}
\begin{align}
\max_{\boldsymbol{\gamma_\mathrm{s}},\boldsymbol{\gamma_\mathrm{r}}}\  \min_{\boldsymbol{\sigma_\mathrm{r}}}\quad & \min\Bigg( \sum_{i=1}^{\min(M,K_\mathrm{r})} \!\log{\!\left(1+
\frac{\sigma_{1_i}^2\gamma_{\mathrm{s}_{\rho{\left(i\right)}}}}{1+\gamma_{\mathrm{r}_i}\sigma_{\mathrm{s}_{\rho{\left(i\right)}}}^2}\right)},\nonumber\\ &{\hspace*{1cm}\sum_{i=1}^{{\min(K_{\mathrm{t}},N)}} \!\log{\!\Big(1+\sigma^2_{2_i} \gamma_{\mathrm{r}_i}\Big)} \Bigg)} \tag{\ref{P:FDbA}}\\
\text{subject to}\quad & \|\boldsymbol{\gamma}_\mathrm{s}\|_1= P_\mathrm{s},\\
&\|\boldsymbol{\gamma}_{r}\|_1 \leq P_{r},\\
&\|\boldsymbol{\sigma}^2_{r}\|_1= T,\\
&\sigma_{1_i}^2\gamma_{\mathrm{s}_{\rho{\left(i\right)}}} \geq \sigma_{1_{i+1}}^2\gamma_{\mathrm{s}_{\rho\left(i+1\right)}},\ \forall i \leq \min(M,K_\mathrm{r}),\label{P:FDbA:ConsA}\\
&\gamma_{\mathrm{r}_i}\sigma_{\mathrm{s}_{\rho{\left(i\right)}}}^2 \geq \gamma_{\mathrm{r}_{i+1}}\sigma_{\mathrm{r}_{\rho\left(i+1\right)}}^2,\ {\forall i \leq \min(K_{\mathrm{t}},N).}\label{P:FDbA.ConsB}
\end{align}
\end{subequations}
%
%
%
%
%
\begin{algorithm}
\caption{Robust Transceiver Design}
\begin{algorithmic}[1]
\State Determine $\boldsymbol\gamma^{(l)}_{r}=[\tau^{(l)}_{r}-\frac{1}{\boldsymbol{\sigma}^2_2}]^{+}$, s.t. $\|\boldsymbol\gamma^{(l)}_{r}\|_1=\bar{P}_{r}^{(l)}$
\State Determine $\bar{\boldsymbol\gamma}^{(l)}_{r}=\boldsymbol\gamma^{(l)}_{r}(1:\min(M,K_\text{t}))$
\State Define $E_0 = 1$ and $\boldsymbol{\sigma^{(0)^2}}_{r} = 0$
\While {$E_2$ large}

\State Obtain $\boldsymbol{\sigma}^{(q)^2}_{r}$
\State Obtain $\boldsymbol\gamma^{(q)}_{s}$, using Algorithm 2

\State$E_2 = \boldsymbol{\sigma}^{(q)^2}_{r} - \boldsymbol{\sigma}^{(q-1)^2}_{r}$
\EndWhile
\State Calculate $R^{(0)}_\mathrm{sr}=\sum_{i=1}^{M} \log_2{\left(1+\bar{v}_{i}\bar{\gamma}_{{s}_i}\right)} $
\State Calculate $R^{(0)}_\mathrm{rd}=\sum_{j=1}^{K_\mathrm{t}} \log_2{\left(1+\sigma^2_{2_i}{\gamma}_{{r}_j}\right)} $
\If {$R^{(0)}_\mathrm{sr}<R^{(0)}_\mathrm{rd}$}
\State Define $U = P_{r}$, $L = 0$, $\bar{P}_{r}^{(1)}=\frac{P_{r}}{2}$ and $E_1 = 1$
\While {$E_1 \geq \frac{\varepsilon}{2}$}
\State Determine $\boldsymbol\gamma^{(l)}_{r}=[\tau^{(l)}_{r}-\frac{1}{\boldsymbol{\sigma}^2_2}]^{+}$, s.t. $\|\boldsymbol\gamma^{(l)}_{r}\|_1=\bar{P}_{r}^{(l)}$
\State Determine $\bar{\boldsymbol\gamma}^{(l)}_{r}=\boldsymbol\gamma^{(l)}_{r}(1:\min(M,K_\text{t}))$
\State Define $E_2 = 1$ and $\boldsymbol{\sigma^{(0)^2}}_{r} = 0$
\While {$E_2$ large}

\State Obtain $\boldsymbol{\sigma}^{(q)^2}_{r}$
\State Obtain $\boldsymbol\gamma^{(q)}_{s}$, using Algorithm 2

\State$E_2 = \boldsymbol{\sigma}^{(q)^2}_{r} - \boldsymbol{\sigma}^{(q-1)^2}_{r}$
\EndWhile
\State Calculate $R^{(l)}_\mathrm{sr}=\sum_{i=1}^{M} \log_2{\left(1+\bar{v}_{i}\bar{\gamma}_{{s}_i}\right)} $
\State Calculate $R^{(l)}_\mathrm{rd}=\sum_{j=1}^{K_\mathrm{t}} \log_2{\left(1+\sigma^2_{2_i}{\gamma}_{{r}_j}\right)} $
\If{$R^{(l)}_\mathrm{sr}>R^{(l)}_\mathrm{rd}$}
\State $U = \bar{P}_{r}$
\State $\bar{P}_{r} = \frac{U+L}{2}$
\ElsIf{$R^{(l)}_\mathrm{sr}<R^{(l)}_\mathrm{rd}$}
\State $L = \bar{P}_{r}$
\State $\bar{P}_{r} = \frac{U+L}{2}$
\EndIf
\State $E_1 = |U-L|$
\EndWhile 
\EndIf
\end{algorithmic}
\label{alg:MaterialCharacterization}
\end{algorithm}

\color{black}
\begin{algorithm}
\caption{The optimal $\boldsymbol{\gamma}_{s}$}
\begin{algorithmic}[1]
\State Find power allocation $\boldsymbol{P}$ using well-known regular water-filling algorithm 
\While {temp is large}
\State $\text{temp}=0$
\For{$i$}
\If{$P_i>\text{cap}_i$}
\State $P_i=\text{cap}_i$
\State $\text{temp}=\text{temp}+P_i-\text{cap}_i$
\EndIf
\EndFor
\State $\boldsymbol{P}=\boldsymbol{P}+\frac{\text{temp}}{\text{number of channels}}$
\EndWhile 
\end{algorithmic}
\label{alg:MaterialCharacterization}
\end{algorithm}
Now, we need to solve the optimization problem~\eqref{P:FDbA}. It can be readily shown that $R_{\text{rd}}^{\text{FD}}$ is monotonically increasing function of $P_r$. Furthermore, one can show that $R_{\text{sr}}^{\text{FD}}$ is an increasing function w.r.t. $P_s$ and decreasing function w.r.t. $T$ and $P_r$ (See Appendix III). Consequently, the worst-case RSI chooses a strategy to reduce the spectral efficiency, while the relay and the source cope with such strategy for improving the system robustness. That means, on one hand the RSI hurts the stronger eigendirections of the received signal space more than the weaker ones. However, on the other hand the source tries to cope with this strategy adaptively by smart eigen selection.
This process clearly makes optimization problem complicated at the source-relay side. Unlike the source-relay side, the resource allocation problem at the relay-receiver side is rather an easy task. Because in the relay-receiver side there is only one maximization and we can find the sum capacity simply by using the well-known water-filling algorithm. Observe that although finding each of $R_{\text{sr}}^{\text{FD}}$ and $R_{\text{rd}}^{\text{FD}}$ separately is a convex problem, the problem~\eqref{P:FDbA} as a whole is a non-convex one. Therefore in this paper, first we try to find each of $R_{\text{sr}}^{\text{FD}}$ and $R_{\text{rd}}^{\text{FD}}$ separately by doing a convex optimization problem, and then, assign the best values to $R_{\text{sr}}^{\text{FD}}$ and $R_{\text{rd}}^{\text{FD}}$ accordingly in order to reach the global optimum. As a result, the algorithm suggested by this paper gives an achievable rate and not the exact capacity. Notice that the optimum values for the transmission power on relay side may not sum to $\mathrm{P}_{r}$. the reason is that $R^{\mathrm{FD}}_\mathrm{sr}$ is a monotonically decreasing function of $\mathrm{P}_{r}$ and as we are interested in the $\min (R^{\mathrm{FD}}_\mathrm{sr},R^{\mathrm{FD}}_\mathrm{rd})$, in the case of $R^{\mathrm{FD}}_\mathrm{sr}<R^{\mathrm{FD}}_\mathrm{rd}$ we will have $\min (R^{\mathrm{FD}}_\mathrm{sr},R^{\mathrm{FD}}_\mathrm{rd}) = R^{\mathrm{FD}}_\mathrm{sr}$. Therefore it is in our interest not to use all the allowed power at the relay transmitter to increase $R^{\mathrm{FD}}_\mathrm{sr}$. Similarly, in the case of $R^{\mathrm{FD}}_\mathrm{sr}>R^{\mathrm{FD}}_\mathrm{rd}$ we have $\min (R^{\mathrm{FD}}_\mathrm{sr},R^{\mathrm{FD}}_\mathrm{rd}) = R^{\mathrm{FD}}_\mathrm{rd}$ which can be increased by increasing the total power usage of relay's transmitter.  In general, it is easy to check that in order for $\min (R^{\mathrm{FD}}_\mathrm{sr},R^{\mathrm{FD}}_\mathrm{rd})$ to be maximized, one must have $R^{\mathrm{FD}}_\mathrm{sr}=R^{\mathrm{FD}}_\mathrm{rd}$. However, in the case that we are already using the maximum allowed power at the relay's transmitter and we still have $R^{\mathrm{FD}}_\mathrm{sr}>R^{\mathrm{FD}}_\mathrm{rd}$, the problem cannot be further improved and the algorithm ends. The general idea of our algorithm is to solve $R^{\mathrm{FD}}_\mathrm{sr}$ and $R^{\mathrm{FD}}_\mathrm{rd}$ separately and then try to find $\mathrm{P}^{\star}_{r}$ in a way we have $R^{\mathrm{FD}}_\mathrm{sr}(\mathrm{P}^{\star}_{r})=R^{\mathrm{FD}}_\mathrm{rd}(\mathrm{P}^{\star}_{r})$. If it turns out that $\mathrm{P}^{\star}_{r}>\mathrm{P}^{}_{r}$, this means that the best value for relay input power is beyond the limit imposed by the constraints, so we consider $\mathrm{P}^{}_{r}$ as the optimal value for the relay's input power. In order to find $\mathrm{P}^{\star}_{r}$, first we define the function $g(\mathrm{P}^{}_\mathrm{})=R^{\mathrm{FD}}_\mathrm{sr}(\mathrm{P}^{}_\mathrm{})-R^{\mathrm{FD}}_\mathrm{rd}(\mathrm{P}^{}_\mathrm{})$ as a function of relay input power $\mathrm{P}^{}_\mathrm{}$, and then we find each $R^{\mathrm{FD}}_\mathrm{sr}(\mathrm{P}^{}_\mathrm{})$ and $R^{\mathrm{FD}}_\mathrm{rd}(\mathrm{P}^{}_\mathrm{})$ separately. To find $g(\mathrm{P})$ for every given $\mathrm{P}$, first we use water filling and find the best input power $\boldsymbol{\lambda}_{r}$ policy at the relay-receiver side. Then $R^{\mathrm{FD}}_\mathrm{rd}(\mathrm{P}^{}_{r})$ can be readily calculated. Next, we use $\boldsymbol{\lambda}_{r}$ to calculate $\boldsymbol{\lambda}_{s}$ and $\boldsymbol{\sigma}_{r}$ at the transmitter-ralay side. Finally, having all $\boldsymbol{\lambda}_{r}$, $\boldsymbol{\lambda}_{s}$ and $\boldsymbol{\sigma}_{r}$ we can find the value of  $R^{\mathrm{FD}}_\mathrm{sr}(\mathrm{P}^{}_{r})$. As discussed before, we are interested in finding the $\mathrm{P}$ such that $R^{\mathrm{FD}}_\mathrm{sr}(\mathrm{P})=R^{\mathrm{FD}}_\mathrm{rd}(\mathrm{P})$. Observe that, to find such $\mathrm{P}$ it is sufficient to find the zeros of $g$. Hence, in the main algorithm we first solve the problem by setting $\mathrm{P}^{}_\mathrm{}=\mathrm{P}^{}_{r}$ and then check whether we have 
$R^{\mathrm{FD}}_\mathrm{sr}<R^{\mathrm{FD}}_\mathrm{rd}$ or $R^{\mathrm{FD}}_\mathrm{sr}\geq R^{\mathrm{FD}}_\mathrm{rd}$. In case that $R^{\mathrm{FD}}_\mathrm{sr}\geq R^{\mathrm{FD}}_\mathrm{rd}$ happens, the algorithm ends since it suggests $\mathrm{P}^{\star}_{r}\geq\mathrm{P}^{}_{r}$, and by taking power constraints into account we conclude $\mathrm{P}^{\star}_{r}=\mathrm{P}^{}_{r}$. Otherwise, algorithm should keep going until it finds $\mathrm{P}^{\star}_{r}$. Note that when $R^{\mathrm{FD}}_\mathrm{sr}<R^{\mathrm{FD}}_\mathrm{rd}$ we have $g(\mathrm{P}^{}_{r})<0$. Also we know $g(0)=R^{\mathrm{FD}}_\mathrm{rd}(0)-R^{\mathrm{FD}}_\mathrm{sr}(0)=R^{\mathrm{FD}}_\mathrm{rd}(0)\geq0$. Therefore, for $\mathrm{P}\in [0 \ \ \mathrm{P}_{r}]$ we have $g(a)g(b)\leq 0$ and we can use the well-known Bisection method to find the zero of $g$ (Bisection can find the zero of a continuous function $g$ in the interval $[a \ \ b]$ if we have $g(a)g(b)\leq 0$). Also, as mentioned before $R^{\mathrm{FD}}_{\mathrm{sr}}$ is a monotonically decreasing function of $\mathrm{P}_{r}$ whereas $R^{\mathrm{FD}}_{\mathrm{rd}}$ is monotonically increasing function of $\mathrm{P}_{r}$ (See Appendix III). As a result, $g$ is a monotonically decreasing function of $\mathrm{P}^{}_{r}$ and therefore, it only has one zero in $[0 \ \ \mathrm{P}_{r}]$. The pseudo code for for finding the optimal singular and eigenvalues is provided in Algorithm 1.

Now we focus on how to find $R^\mathrm{FD}_\mathrm{sr}$. In order to find the sum rate for source-relay part, we assume that we are already given $\boldsymbol{\gamma}^\star_{r}$ which is the vector of relay input powers that maximizes the sum rate in relay-destination part. The next step is to do the minimization over $\boldsymbol{\sigma}_{r}$ and the maximization over $\boldsymbol{\gamma}_{s}$. One approach to solve this problem is to solve it iteratively. In this method, first one finds the optimal $\boldsymbol{\gamma}_{s}$ by solving the maximization part under the assumption that the optimal $\boldsymbol{\sigma}_{r}$ is given, and then, having the optimal $\boldsymbol{\gamma}_{s}$ the minimization problem can be solved efficiently. This process goes on until the convergence of $\boldsymbol{\gamma}_{s}$ and/or $\boldsymbol{\sigma}_{r}$. The maximization part is done using water-filling method. However, the extra conditions $\forall \ i \leq \min(M,K_{{r}}), \sigma_{1_i}^2\gamma_{{s}_{\rho{\left(i\right)}}} \geq \sigma_{1_{i+1}}^2\gamma_{{s}_{i+1}}$ should be taken into account. For instance, if the optimal value for $\gamma_{s_i}$ turns out to be equal to zero, we then should have $\gamma_{s_j}=0$ for all $j>i$ irrespective of their SNR. As it can be seen in Fig. \ref{fig:WaterLevel} these extra restrictions act like caps on top of the water and creates  multilevel water-filling which can be construed as a cave inside the water. Algorithm 2 provides the detail of multilevel water-filling. For the minimization part, Lagrangian multiplier is used. We have
\begin{align}
    L = \sum_{i=1}^{\min(M,K_{r})}& \!\log_2{\!\left(1+
\frac{\sigma_{1_i}^2\gamma_{{s}_{\rho{\left(i\right)}}}}{1+\gamma_{{r}_i}\sigma_{{r}_{\rho{\left(i\right)}}}^2}\right)} + \lambda \left( \sum_{i=0}^{N}\sigma^2_{r_i} - T\right)
\end{align}
Calculating $\frac{\partial L}{\partial \sigma^2_{r_i}}=0$ we arrive at
{\begin{align}
\sigma^2_{r_i} = \left[\frac{\pm\sqrt{\left(\sigma^2_{1_i}\gamma_{s_i}\right)^2+\frac{4\sigma^2_{1_i}\gamma_{s_i}\gamma{r_i}}{\lambda}}-\sigma^2_{1_i}\gamma_{s_i}-2}{2\gamma_{r_i}}\right]^+.
\end{align}}
As $\sigma^2_{r_i}$ is always non-negative, the only solution would be
n{\begin{align}
\sigma^2_{r_i} = \left[\frac{\sqrt{\left(\sigma^2_{1_i}\gamma_{s_i}\right)^2+\frac{4\sigma^2_{1_i}\gamma_{s_i}\gamma{r_i}}{\lambda}}-\sigma^2_{1_i}\gamma_{s_i}-2}{2\gamma_{r_i}}\right]^+,\label{eq:condition}
\end{align}}
where $\lambda$ is the water level. Similarly to the maximization case, there is also extra constraints $\gamma_{{r}_i}\sigma_{{s}_{\rho{\left(i\right)}}}^2 \geq \gamma_{{r}_{i+1}}\sigma_{{r}_{i+1}}^2$ that must be considered during the minimization process. However, it can be shown that if the constraints $\gamma_{r_{i}} \geq \gamma_{r_{i+1}}$ and $\sigma_{1_i}^2\gamma_{{s}_{\rho{\left(i\right)}}} \geq \sigma_{1_{i+1}}^2\gamma_{{s}_{i+1}}$ are already met, then constraint $\gamma_{{r}_i}\sigma_{{s}_{\rho{\left(i\right)}}}^2 \geq \gamma_{{r}_{i+1}}\sigma_{{r}_{i+1}}^2$ becomes redundant. The proof is given in Appendix IV.

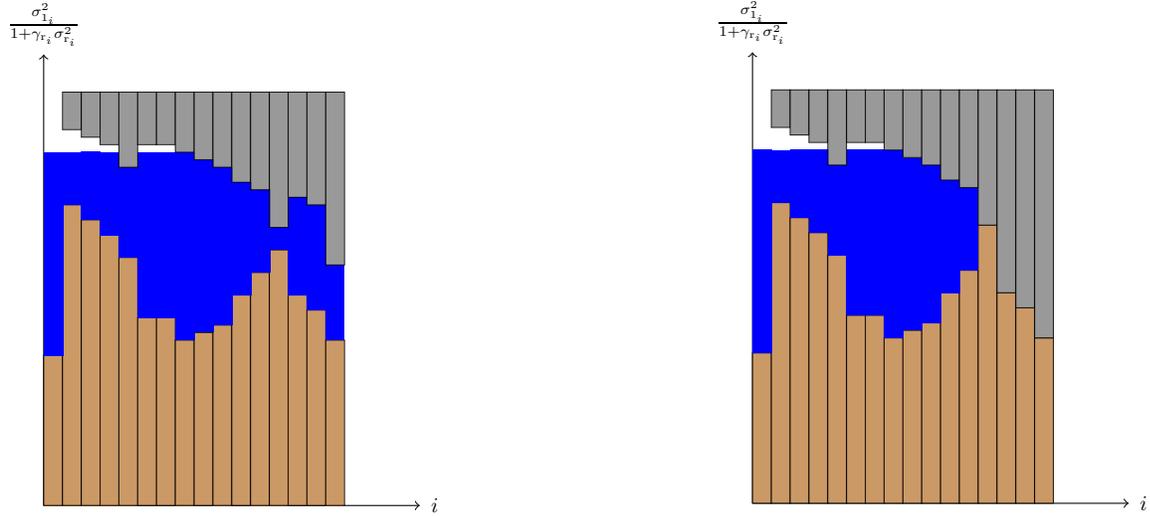
\begin{figure*}
\centering
\quad
\begin{minipage}{0.4\textwidth}
\subfigure[In this case there is no subchannel with optimum power equal to zero. Note that due to the power cap constraint~\eqref{P:FDb:ConsA}, water has not the same level for all subchannels.]{
\tikzset{every picture/.style={scale=1}, every node/.style={scale=0.8}}%
\input{WaterLevel}
}
\end{minipage}\quad\quad\quad\quad\quad\quad
\begin{minipage}{0.4\textwidth}
\subfigure[The case where there is one subchannel with optimum power equal to zero. Notice that in this case, due to the additional power cap constraint~\eqref{P:FDb:ConsA} all the remaining subchannels also get the zero power.]{
\tikzset{every picture/.style={scale=1}, every node/.style={scale=0.8}}%
\input{WaterLevelB}
}
\end{minipage}
\caption{Multilevel water-filling}
\label{fig:WaterLevel}
\end{figure*}
\begin{figure*}
\centering
\begin{minipage}{0.5\textwidth}
\centering
\subfigure[$\{M,K_{r}+K_\mathrm{t},N\}=\{4,10,4\}$]{
\tikzset{every picture/.style={scale=.8}, every node/.style={scale=0.8}}%
\input{myfile}
\label{fig:HdVsFdE}
}
\end{minipage}%
\begin{minipage}{0.5\textwidth}
\centering
\subfigure[$\{M,K_{r}+K_\mathrm{t},N\}=\{10,24,10\}$]{
\tikzset{every picture/.style={scale=.8}, every node/.style={scale=0.8}}%
\input{asli}
\label{fig:HdVsFdE}
}
\end{minipage}
\caption{The transmit power budget at the source and the relay are assumed to be equal, i.e., $P_{s}=P_\mathrm
r=P=5$.}
\label{fig:HdVsFdBD}
\end{figure*}
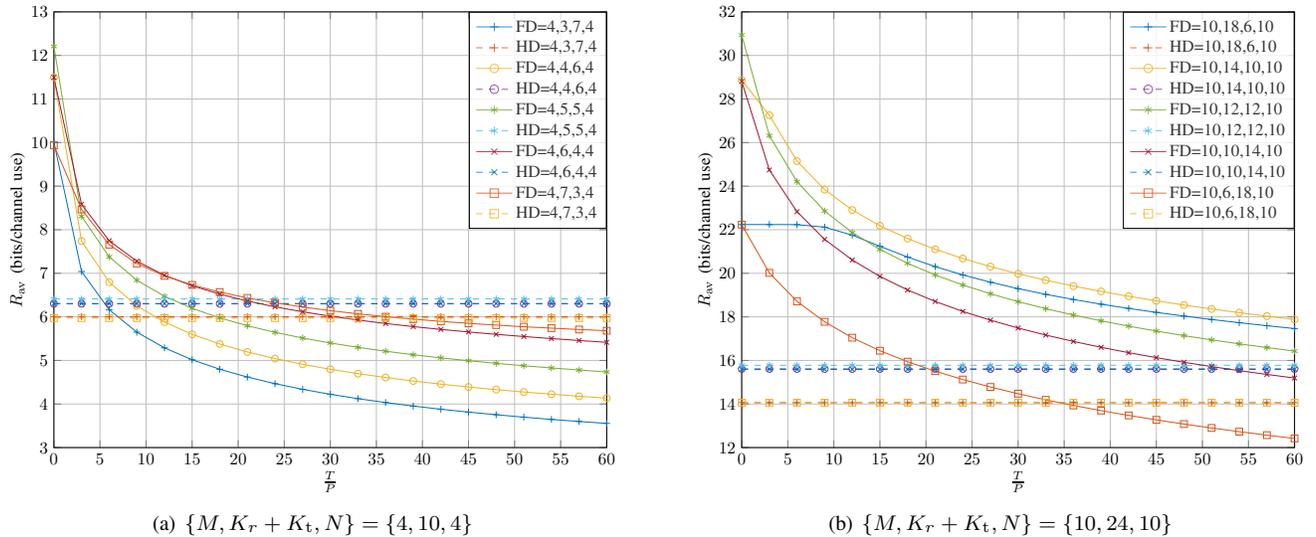
\section{Numerical Results}
We assume equal transmit power budgets at the source and at the relay are equal and we provide the simulations for case $P=P_{s}=P_{r}=5$ . Moreover, the receiver AWGN variance is assumed to be unity. We investigate the performance of full-duplex relaying with RSI channel uncertainty bound $T$, i.e.,  $\mathrm{Tr}(\bar{\mathbf{H}}_{r}\bar{\mathbf{H}}^H_{r})\leq T$. We consider the column vectors of the source-relay and the relay-destination channel matrices to be from zero-mean Gaussian distribution with identity covariance matrices. That means, by representing the $i$-th column of $\mathbf{H}_1$ and $j$-th column of $\mathbf{H}_2$ as $\mathbf{h}_{1i}$ and $\mathbf{h}_{2j}$, respectively, we assume $\mathbf{h}_{1i}\sim\mathcal{CN}(\mathbf{0},\mathbf{I})$ and  $\mathbf{h}_{2j}\sim\mathcal{CN}(\mathbf{0},\mathbf{I})$. We perform Monte-Carlo simulations with $L=10^4$ realizations from random channels and noise vectors. Hence, the average worst-case throughput rate is defined as the average of worst-case rates for $L$ randomizations, i.e., 
$
R_{\mathrm{av}}=\frac{1}{L}\sum_{l=1}^{L}R_l.
$
Notice that, for each set of realizations, i.e., $\{\mathbf{H}_1,\mathbf{H}_2,\mathbf{n}_{r},\mathbf{n}_\mathrm{d}\}$, we solve the robust transceiver design as is elaborated in Algorithm 1. We run two sets of simulations as described in two following subsections.

\subsection{Antenna Array Increment}
We consider two cases, where the source, relay and destination are equipped with (a)- small antenna array and (b)- large antenna arrays. For these cases, we have
\begin{enumerate}[(a)-]
\item $M=4,K_{r}+K_\mathrm{t}=10,N=4$
\item $M=10,K_{r}+K_\mathrm{t}=24,N=10$
\end{enumerate}
These cases are considered to highlight the performance of full-duplex DF relaying as a function of number of antennas with the worst-case RSI. Interestingly, as the number of antennas at the source, relay and destination increase, full-duplex relaying achieves a higher throughput rate even with strong RSI. Furthermore, notice that the worst-case RSI casts strong interference on the strong streams from the source to the destination. With very low RSI power $T\rightarrow 0$, full-duplex almost doubles the throughput rate compared to the half-duplex counterpart. This can be seen in~\figurename{~\ref{fig:HdVsFdBD}}, where the curves cross the vertical axis. However, as $T$ increases, the efficiency of full-duplex operation drops.


\begin{figure*}
\centering
\begin{minipage}{.5\textwidth}
\centering
\subfigure[$\{M,K_\mathrm{t}+K_{r},N\}=\{4,8,4\}$]{
\tikzset{every picture/.style={scale=.8}, every node/.style={scale=0.8}}%
\input{CwrtT2}
\label{fig:CwrtT}
}
\end{minipage}%
\begin{minipage}{.5\textwidth}
\centering
\subfigure[$\{M,K_\mathrm{t}+K_{r},N\}=\{8,8,8\}$]{
\tikzset{every picture/.style={scale=.8}, every node/.style={scale=0.8}}%
\input{CwrtT}
\label{fig:CwrtT1}
}
\end{minipage}
\caption{The transmit power budget at the source and the relay are assumed to be equal, i.e., $P_{s}=P_\mathrm
r=P=5$.}
\label{fig:HdVsFdBD}
\end{figure*}
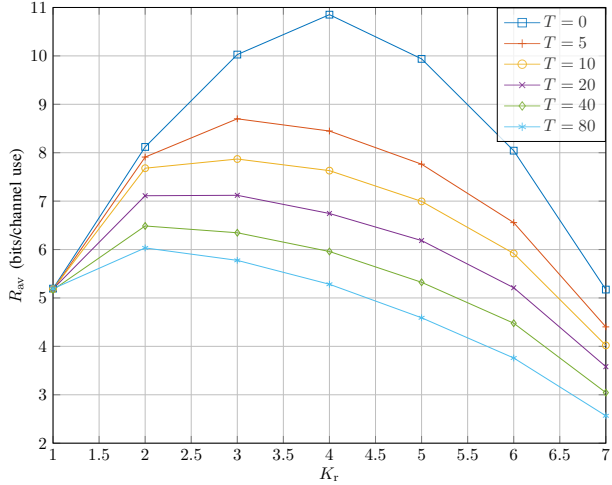
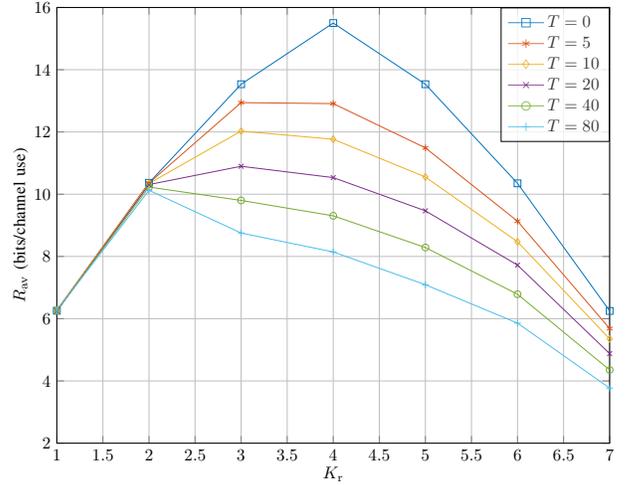
\subsection{Relay Tx/Rx Antenna allocation}
Let the relay have $K_\mathrm{t}+K_{r}=8$ in total. Furthermore, suppose two cases where the number of antenna at the source and relay are $\{M,N\}=\{4,4\}$ and $\{M,N\}=\{8,8\}$ respectively. The question is, from $8$ antennas at the relay, how many should be used for reception for the robust design?. To answer this question, we consider the following scenarios
\begin{enumerate}[(a):]
\item $\{K_\mathrm{t},K_{r}\}=\{1,7\}\Rightarrow
\mathrm{DoF}_\mathrm{sr}=3,\ \mathrm{DoF}_\mathrm{rd}=1$
\item $\{K_\mathrm{t},K_{r}\}=\{2,6\}\Rightarrow
\mathrm{DoF}_\mathrm{sr}=4,\ \mathrm{DoF}_\mathrm{rd}=2$
\item $\{K_\mathrm{t},K_{r}\}=\{3,5\}\Rightarrow
\mathrm{DoF}_\mathrm{sr}=4,\ \mathrm{DoF}_\mathrm{rd}=3$
\item $\{K_\mathrm{t},K_{r}\}=\{4,4\}\Rightarrow
\mathrm{DoF}_\mathrm{sr}=4,\ \mathrm{DoF}_\mathrm{rd}=4$
\item $\{K_\mathrm{t},K_{r}\}=\{5,3\}\Rightarrow
\mathrm{DoF}_\mathrm{sr}=4,\ \mathrm{DoF}_\mathrm{rd}=3$
\item $\{K_\mathrm{t},K_{r}\}=\{6,2\}\Rightarrow
\mathrm{DoF}_\mathrm{sr}=4,\ \mathrm{DoF}_\mathrm{rd}=2$
\item $\{K_\mathrm{t},K_{r}\}=\{7,1\}\Rightarrow
\mathrm{DoF}_\mathrm{sr}=4,\ \mathrm{DoF}_\mathrm{rd}=1$
\end{enumerate}
As can be seen in~\figurename{~\ref{fig:HdVsFdBD}}, by using more antennas for reception than for transmission, i.e., $K_{r}>K_\mathrm{t}$, at the relay, the throughput rate is maximized for worst-case FD relay. This is due to the fact that, increasing the signal-to-noise ratio (SNR) of the source-relay streams enhances the overall throughput rate more than the increase by the $\mathrm{DoF}$ of the relay-destination link. However, notice that in this setup the overall $\mathrm{DoF}$ from the source to destination is limited by the $\mathrm{DoF}$ of the source-relay link.

~\figurename{~\ref{fig:HdVsFdBD}} shows the sum rate as a function of $K_{r}$ for different values of $T$ where $\{M,K_\mathrm{t}+K_{r},N\}=\{4,8,4\}$ and $\{M,K_\mathrm{t}+K_{r},N\}=\{8,8,8\}$ respectively. As it can be seen, when $T=0$, by increasing the total $\mathrm{DoF}$ sum rate increases as well. Furthermore, at a specific total $\mathrm{DoF}$ increasing $T$ is more harmful for cases where $K_\mathrm{r}\leq K_\mathrm{t}$ compared to the cases where $K_{r}>K_\mathrm{t}$. For instance, for the case of $\{M,K_\mathrm{t},K_{r},N\}=\{4,3,5,4\}$ we have $\mathrm{DoF}_\mathrm{sr}=3$ and $\mathrm{DoF}_\mathrm{rd}=4$ which means the total $\mathrm{DoF}$ of the system is 3. Here, the results show that although the total $\mathrm{DoF}$ for both $\{M,K_\mathrm{t},K_{r},N\}=\{4,3,5,4\}$ and $\{M,K_\mathrm{t},K_{r},N\}=\{4,5,3,4\}$ is 3, the sum rate capacity of the latter is much better than that of the former. This is because of the fact that when $\mathrm{DoF}_\mathrm{sr}>\mathrm{DoF}_\mathrm{rd}$, interference can at most damage the SNR of $\mathrm{DoF}_\mathrm{sr}-\mathrm{DoF}_\mathrm{rd}$ sub channels at the source-relay side. Therefore, source can manage to gain more sum rate by choosing its power allocation wisely. On the other hand, in the case of $\mathrm{DoF}_\mathrm{sr}\leq\mathrm{DoF}_\mathrm{rd}$, no matter how well the power allocation is done, all sub channels suffer from interference at the source-relay end.

\section{Conclusion}
In this paper, we investigated a multi-antenna source communicating with a multi-antenna destination through a multi-antenna relay. The relay is assumed to exploit a decode-and-forward (DF) strategy. The transceivers are designed in order to be robust against the worst-case residual self-interference (RSI). To this end, the worst-case achievable throughput rate is maximized. This optimization problem turns out to be a non-convex problem. Assuming that the degrees-of-freedom (DoF) of the source-relay link is less than the DoF of the relay-destination link, we determined the left and right matrices of the singular vectors of the worst-case RSI channel. Then, the problem is simplified to the optimal power allocation at the transmitters, which guarantees robustness against the worst-case RSI singular values. This simplified problem is still non-convex. Based on the intuitions for optimal power allocation at the source and relay, we proposed an efficient algorithm to capture a stationary point. Hence, in a DF relay with multi-stream beamforming, we determine the critical point where the half-duplex relaying outperforms the full-duplex relaying. This critical point provides a mode-switching threshold in hybrid half-duplex full-duplex relay systems. 
\section{appendices}
\section*{Appendix I\\ proof of Theorem 1}\label{App:Ia}
Before stating the proof, first we introduce the following definitions.
\begin{definition}
For a vector $\boldsymbol{a}$, we denote vector $\boldsymbol{a}^{\downarrow}$ which has the same components as $\boldsymbol{a}$ except that they are sorted in a decreasing order.
\end{definition}
\begin{definition}
The vector $\boldsymbol{a}$ is said to be majorized by vector $\boldsymbol{b}$ and denoted by $\boldsymbol{a} \prec \boldsymbol{b}$ if we have:
\begin{align}
    \sum_{i = 1}^{K} a^{\downarrow}_{i} \leq \sum_{i = 1}^{K} b^{\downarrow}_{i},\\
    \sum_{i = 1}^{N} a^{\downarrow}_{i} = \sum_{i = 1}^{N} b^{\downarrow}_{i},
\end{align}
where $a^{\downarrow}_i$ is the $i$'th component of $\boldsymbol{a}^{\downarrow}$, $N$ is the number of vector components and $K \leq N$. If the last equality does not hold, $\boldsymbol{a}$ is said to be weakly majorized by $\boldsymbol{b}$ and denoted by $\boldsymbol{a} \prec_{w} \boldsymbol{b}$
\end{definition}
\begin{definition}
The vector $\boldsymbol{a}$ is said to be multiplicatively majorized by vector $\boldsymbol{b}$ and denoted by $\boldsymbol{a} \prec_{\times} \boldsymbol{b}$ if we have:
\begin{align}
    \prod_{i = 1}^{K} a^{\downarrow}_{i} \leq \prod_{i = 1}^{K} b^{\downarrow}_{i},\\
    \prod_{i = 1}^{N} a^{\downarrow}_{i} = \prod_{i = 1}^{N} b^{\downarrow}_{i}.
\end{align}
also, it is easy to check
\begin{equation}
\boldsymbol{a} \prec_{\times} \boldsymbol{b} \quad \Leftrightarrow \quad \log(\boldsymbol{a}) \prec \log(\boldsymbol{b})
\end{equation}
\end{definition}
To begin with, we know that for $n \times m$ matrix $\boldsymbol{A}$ and $m \times n$ matrix $\boldsymbol{B}$ we have $\lambda_i(\boldsymbol{A}\boldsymbol{B}) = \lambda_i(\boldsymbol{B}\boldsymbol{A}),\ \forall i\in\{1,\cdots,\min(m,n)\}$. Also the only difference between eigenvalues of $\boldsymbol{B}\boldsymbol{A}$ and $\boldsymbol{A}\boldsymbol{B}$ are the number of eigenvalues 0. Thus, non-zero eigenvalues of  $\bar{\mathbf{H}}_{r}\mathbf{Q}_{r}\bar{\mathbf{H}}^H_{r}$ and $\mathbf{Q}_{r}\bar{\mathbf{H}}^H_{r}\bar{\mathbf{H}}_{r}$ and also ${\mathbf{H}}_\mathrm{1}\mathbf{Q}_{s}{\mathbf{H}}^H_\mathrm{1}$ and $\mathbf{Q}_{s}{\mathbf{H}}^H_\mathrm{1}{\mathbf{H}}_\mathrm{1}$ are equal, respectively. Notice that all $\mathbf{Q}_{s}$,  ${\mathbf{H}}^H_\mathrm{1}{\mathbf{H}}_\mathrm{1}$, $\mathbf{Q}_{r}$ and $\bar{\mathbf{H}}^H_{r}\bar{\mathbf{H}}_{r}$ are square matrices. For ${\mathbf{H}}^H_\mathrm{1}{\mathbf{H}}_\mathrm{1}$ and $\bar{\mathbf{H}}^H_{r}\bar{\mathbf{H}}_{r}$ we define
$\lambda_i({\mathbf{H}}^H_\mathrm{1}{\mathbf{H}}_\mathrm{1}) = \sigma_i^{2}({\mathbf{H}}_\mathrm{1})$ and $\lambda_i(\bar{\mathbf{H}}^H_{r}\bar{\mathbf{H}}_{r}) =\sigma_i^{2}(\bar{\mathbf{H}}_{r})$ respectively.

As discussed in~Remark\ref{Remark1}, the equality $\lambda_i(\mathbf{Q}_{s}{\mathbf{H}}^H_\mathrm{1}{\mathbf{H}}_\mathrm{1}) = \lambda_{\rho\left(i\right)}(\mathbf{Q}_{s})\lambda_i({\mathbf{H}}^H_\mathrm{1}{\mathbf{H}}_\mathrm{1})$ does not hold in general. However, using the definition of determinant one can arrive at the following equality,
\begin{equation}
    \prod_{i = 1}^{\text{min}(M, K_r)} \lambda_i(\mathbf{Q}_{s}{\mathbf{H}}^H_\mathrm{1}{\mathbf{H}}_\mathrm{1}) = \prod_{i=1}^{\text{min}(M, K_r)} \lambda_i(\mathbf{Q}_{s})\sigma^2_i({\mathbf{H}}_\mathrm{1}).
\end{equation}
Now, we define vector $\boldsymbol{\lambda}(\mathbf{Q}'_{s})$ and set its components to be ${\lambda_{\rho{\left(i\right)}}}(\mathbf{Q}'_{s}) = \frac{\lambda_i(\mathbf{Q}_{s}{\mathbf{H}}^H_\mathrm{1}{\mathbf{H}}_\mathrm{1})}{\sigma^2_i({\mathbf{H}}_\mathrm{1})}$. By defining $\lambda_{\rho\left(i\right)}(\mathbf{Q}'_{s})$ instead of $\lambda_{i}(\mathbf{Q}'_{s})$, we emphasize that the elements of $\boldsymbol{\lambda}(\mathbf{Q}'_{s})$ are not necessarily in decreasing order. 
Then, we construct the matrix $\mathbf{Q}'_{s}$ having the same eigenvectors as those of ${\mathbf{H}}^H_\mathrm{1}{\mathbf{H}}_\mathrm{1}$ and the eigenvalues $\lambda_{\rho\left(i\right)}(\mathbf{Q}'_{s})$. One can check that for each $i$ we have $\lambda_i(\mathbf{Q}'_{s}{\mathbf{H}}^H_\mathrm{1}{\mathbf{H}}_\mathrm{1})=\lambda_i(\mathbf{Q}_{s}{\mathbf{H}}^H_\mathrm{1}{\mathbf{H}}_\mathrm{1})$. Also, by the definition of $\lambda_{\rho\left(i\right)}(\mathbf{Q}'_{s})$ we have:
\begin{align}
       &\lambda_{\rho\left(i\right)}(\mathbf{Q}'_{s}) = \frac{\lambda_i(\mathbf{Q}_{s}{\mathbf{H}}^H_\mathrm{1}{\mathbf{H}}_\mathrm{1})}{\sigma^2_i({\mathbf{H}}_\mathrm{1})} ,\\ 
       \Rightarrow &\!\log{\!\left(\lambda_{\rho\left(i\right)}(\mathbf{Q}'_{s})\!\right)}\!=\!\log\left( \lambda_i(\mathbf{Q}_{s}{\mathbf{H}}^H_\mathrm{1}{\mathbf{H}}_\mathrm{1})\right)\!-\!\log\!\left(\sigma^2_i({\mathbf{H}}_\mathrm{1})\right),\\ 
       \Rightarrow &\log{\left(\boldsymbol{\lambda}(\mathbf{Q}'_{s})\right)}=\log\left( \boldsymbol{\lambda}(\mathbf{Q}_{s}{\mathbf{H}}^H_\mathrm{1}{\mathbf{H}}_\mathrm{1})\right) - \log\left(\boldsymbol{\sigma}^2({\mathbf{H}}_\mathrm{1})\right).
\end{align}

\begin{lemma}
If $\boldsymbol{A}$ and $\boldsymbol{B}$ are semidefinite Hermitian matrices with $\lambda_{\min(m,n)}(\boldsymbol{AB}) > 0$, then
\begin{equation}
\log(\boldsymbol{\lambda}(\boldsymbol{AB})) -\log(\boldsymbol{\lambda}(\boldsymbol{B})) \prec \log(\boldsymbol{\lambda}(\boldsymbol{A})).
\end{equation}
\end{lemma}
\begin{proof}
The proof is given in \cite[H.1,e]{marshall1979inequalities}.
\end{proof}
\noindent Using the above lemma we can conclude
\begin{equation}
\log\left(\boldsymbol{\lambda}(\mathbf{Q}'_{s})\right) \prec \log\left(\boldsymbol{\lambda}(\mathbf{Q}_{s})\right).
\end{equation}

\noindent Then, immediately we can conclude
\begin{equation}
\boldsymbol{\lambda}(\mathbf{Q}'_{s}) \prec_{\times} \boldsymbol{\lambda}(\mathbf{Q}_{s}).
\end{equation}

\begin{remark}
It is worth mentioning that, depending on channel realizations, the optimal $\mathbf{Q}_{s}$ might contain some zero eigenvalues. In such cases, we can simply ignore the zeros and construct matrix $\mathbf{Q}'_{s}$ with dimension $(n-k) \times (n-k)$ matrix . Similarly, in the cases where ${\bar{\mathbf{H}}}^H_{r}{\bar{\mathbf{H}}_{r}}$ has some zero eigenvalues, we can do the same and proceed to constitute $\bar{\mathbf{H}}'_{s}$ using only nonzero eigenvalues of ${\bar{\mathbf{H}}}^H_{r}{\bar{\mathbf{H}}_{r}}$ and add the zeros back to the result again at the end.
\end{remark}
Finally, we use the following lemma to show that $\bar{\mathbf{H}}'_{\mathrm{r}}$ and $\mathbf{Q}'_{\mathrm{s}}$ are in the feasible set.
\begin{lemma}
For two vectors $\boldsymbol{a}$ and $\boldsymbol{b}$, if we have $\boldsymbol{a} \prec_{\times} \boldsymbol{b}$, then we have $\boldsymbol{a} \prec_{w} \boldsymbol{b}$.
\end{lemma}
\begin{proof}
The proof if given in~\cite[5.A.2.b]{marshall1979inequalities}.
\end{proof}

Exploiting the above lemma one concludes:

\begin{align}
\boldsymbol{\lambda}(\mathbf{Q}'_{s}) \prec_{\times} \boldsymbol{\lambda}(\mathbf{Q}_{s})\ \Rightarrow \  \boldsymbol{\lambda}(\mathbf{Q}'_{s}) \prec_{w} \boldsymbol{\lambda}(\mathbf{Q}_{s}),
\end{align}
which consequently results in
\begin{align}
\sum_{i = 1}^{N}{\lambda_i}(\mathbf{Q}'_{s}) \leq \sum_{i = 1}^{N}{\lambda_i}(\mathbf{Q}_{s}) \Rightarrow  \text{Tr} (\mathbf{Q}'_{s})\leq\text{Tr}(\mathbf{Q}_{s}).
\end{align}
Therefore, there exists $\mathbf{Q}'_{s}$ and $\bar{\mathbf{H}}'_{r}$ fulfilling \eqref{P:newo1a}-\eqref{P:newo1c}, which satisfy
\begin{align}
       \sum_{i = 1}^{\min{(M,K_r)}} \log_2 &\left( 1+ \frac{\lambda_i\left(\mathbf{H}_1\mathbf{Q}_{s}\mathbf{H}^H_1\right)}{1+\lambda_i\left(\bar{\mathbf{H}}_{r}\mathbf{Q}_{r}\bar{\mathbf{H}}^H_{r}\right)} \right) = \\ &\sum_{i = 1}^{\min{(M,K_r)}} \log_2 \left( 1+ \frac{\lambda_{\rho(i)}(\mathbf{Q}'_{s})\sigma^2_i({\mathbf{H}}_\mathrm{1})}{1+\lambda_{i}(\mathbf{Q}_{r})\sigma^2_{\rho(i)}({\bar{\mathbf{H}}'}_\mathrm{r})} \right). 
\end{align}

\section*{Appendix II\\ proof of proposition 2}
In this section we prove that the problem 
\begin{subequations}\label{P:FDc}
\begin{align}
\max_{\boldsymbol{\gamma}_{s},\boldsymbol{\gamma}_{r}}\quad  \min_{\boldsymbol{\sigma}_{{r}}}\quad & \sum_{i=1}^{\min(M,K_{r})} \log_2{\left(1+
\frac{\sigma_{1_i}^2\gamma_{{s}_{\rho{\left(i\right)}}}}{1+\gamma_{{r}_i}\sigma_{{s}_{\rho{\left(i\right)}}}^2}\right)} \tag{\ref{P:FDc}}\\
\text{s.t.}\quad & \|\boldsymbol{\gamma}_{s}\|_1\leq P_{s},\label{P:FDc:ConsA}\\ 
&\|\boldsymbol{\gamma}_{r}\|_1\leq P_{r},\label{P:FDc:ConsB}\\
&\|\boldsymbol{\sigma}^2_{r}\|_1\leq T,\label{P:FDc:ConsC}\\
& \sigma_{1_i}^2\gamma_{{s}_{{\rho{\left(i\right)}}}} \geq \sigma_{1_{i+1}}^2\gamma_{{s}_{{\rho{\left(i+1\right)}}}},\ \forall i \leq \min(M,K_{{r}}),\label{appP:FDb:ConsA}\\
& \gamma_{{r}_i}\sigma_{{r}_{{\rho{\left(i\right)}}}}^2 \geq \gamma_{{r}_{i+1}}\sigma_{{r}_{{\rho{\left(i+1\right)}}}}^2,\ {\forall i \leq \min(K_{\mathrm{t}},N).}\label{appP:FDb.ConsB}
\end{align}
\end{subequations} 
can be further simplified to 
\begin{subequations}\label{P:FDd}
\begin{align}
\max_{\boldsymbol{\gamma}_{s},\boldsymbol{\gamma}_{r}}\quad  \min_{\boldsymbol{\sigma}_{{r}}}\quad & \sum_{i=1}^{\min(M,K_{r})} \log_2{\left(1+
\frac{\sigma_{1_i}^2\gamma_{{s}_{\rho{\left(i\right)}}}}{1+\gamma_{{r}_i}\sigma_{{s}_{\rho{\left(i\right)}}}^2}\right)} \tag{\ref{P:FDd}}\\
\text{s.t.}\quad & \|\boldsymbol{\gamma}_{s}\|_1 = P_{s},\label{P:FDd:ConsA}\\ 
&\|\boldsymbol{\gamma}_{r}\|_1\leq P_{r},\label{P:FDd:ConsB}\\
&\|\boldsymbol{\sigma}^2_{r}\|_1 = T.\label{P:FDd:ConsC}\\
& \sigma_{1_i}^2\gamma_{{s}_{{\rho{\left(i\right)}}}} \geq \sigma_{1_{i+1}}^2\gamma_{{s}_{{\rho{\left(i+1\right)}}}},\ \forall i \leq \min(M,K_{{r}}),\label{app2P:FDb:ConsA}\\
& \gamma_{{r}_i}\sigma_{{r}_{{\rho{\left(i\right)}}}}^2 \geq \gamma_{{r}_{i+1}}\sigma_{{r}_{{\rho{\left(i+1\right)}}}}^2,\ {\forall i \leq \min(K_{\mathrm{t}},N).}\label{app2P:FDb.ConsB}
\end{align}
\end{subequations}
The proof is by contradiction. Starting with the minimization, assume that the optimal vector $\boldsymbol{\sigma^\star}^2_{r}$, for which we have $R^{\text{FD}}_{\text{sr}}(\boldsymbol{\sigma^\star}^2_{r}) \leq R^{\text{FD}}_{\text{sr}}(\boldsymbol{\sigma}^2_{r})$, does not sum to $T$ and thus, we have $\|\boldsymbol{\sigma^\star}^2_{r}\|_1 < T$. Then there exists $\varepsilon > 0$ for which we have $\|\boldsymbol{\sigma^\star}^2_{r}\|_1 +\varepsilon = T$. Now define
\begin{align}
    \varepsilon_i = \frac{\varepsilon \frac{\sigma^2_{1_i}}{\gamma_{r_i}}}{\sum_{j}\frac{\sigma^2_{1_{j}}}{\gamma_{r_{j}}}}.
\end{align}
Note that we have 
\begin{align}
    \sum_{i}\varepsilon_i = \varepsilon, \quad \varepsilon_i \geq 0.
\end{align}
Also, as we have $\varepsilon > 0$, there is at least one $\varepsilon_i$ which is strictly greater than zero i.e. $\varepsilon_i > 0$. 
Now define 
\begin{align}
    \sigma'^2_{r_{\rho\left(i\right)}}= {\sigma^\star}^2_{r_{\rho\left(i\right)}}+\varepsilon_i.
\end{align}
One can check that $\sum_{i} \sigma'^2_{r_{\rho\left(i\right)}} = T$ and $\gamma_{{r}_i}{\sigma'}_{{r}_{{\rho{\left(i\right)}}}}^2 \geq \gamma_{{r}_{i+1}}{\sigma'}_{{r}_{{\rho{\left(i+1\right)}}}}^2,\ {\forall i \leq \min(K_{\mathrm{t}},N)}$. As a result $\boldsymbol{\sigma}'^2_{r_{\rho\left(i\right)}}$ meets the constraints and could be a feasible solution. Note that, as $\boldsymbol{\gamma}^\star_\mathrm{s}$ is the optimal source power allocation based on all other parameters, by changing $\boldsymbol{\sigma^\star}^2_{r}$ to $\boldsymbol{\sigma}^2_{r}$, $\boldsymbol{\gamma}^\star_{s}$ might also change. However, we created each ${\sigma}^2_{r_{\rho\left(i\right)}}$ in a special way to avoid this change. To show this, first notice that we have
\begin{align}
    \gamma^\star_{s_{\rho\left(i\right)}} &= \left[\lambda - \frac{1+\gamma_{r_i}{\sigma^\star}^2_{r_{\rho\left(i\right)}}}{\sigma^2_{1_i}}\right]^+,
\end{align}
where $\lambda$ is water level and can be found based on power constraints.
Substituting new power allocation for interference, we get new power allocation for input power as follow
\begin{align}
    \gamma_{s_{\rho\left(i\right)}} &= \left[\lambda - \frac{1+\gamma_{r_i}{\sigma'}^2_{r_{\rho\left(i\right)}}}{\sigma^2_{1_i}}\right]^+\\
    & = \left[\lambda - \frac{1+\gamma_{r_i}({\sigma^\star}^2_{r_{\rho\left(i\right)}}+\varepsilon_i)}{\sigma^2_{1_i}}\right]^+\\
    & = \left[\lambda - \frac{1+\gamma_{r_i}{\sigma^\star}^2_{r_{\rho\left(i\right)}}}{\sigma^2_{1_i}}+\frac{\varepsilon}{\sum_{j=1}^{N}\frac{\sigma^2_{1_{j}}}{\gamma_{r_{j}}}}\right]^+\\
    &\stackrel{(a)}{=}\left[\lambda' - \frac{1+\gamma_{r_i}{\sigma^\star}^2_{r_{\rho\left(i\right)}}}{\sigma^2_{1_i}}\right]^+\\
    &= \gamma^\star_{s_{\rho\left(i\right)}},
\end{align}
where $(a)$ comes from the fact that $\sfrac{\varepsilon}{\sum_{j=1}^{N}\frac{\sigma^2_{1_{j}}}{\gamma_{r_{j}}}}$ is a constant independent of $i$, so we can define $\lambda'= \lambda+\sfrac{\varepsilon}{\sum_{j=1}^{N}\frac{\sigma^2_{1_{j}}}{\gamma_{r_{j}}}}$. This shows, for $\boldsymbol{\sigma'}^2_{r}$, all the optimal variables and parameters remain the same as those of $\boldsymbol{\sigma^\star}^2_{r}$. Now we compare $R^{\text{FD}}_{\text{sr}}$ for both cases. First, notice that we have $\forall i, \varepsilon_i \geq 0$ and among them there is at least one index $i'$, for which we have $\varepsilon_{i'}>0$. This means $\forall i, {\sigma'}^2_{r_{\rho\left(i\right)}} \geq {\sigma^\star}^2_{r_{\rho\left(i\right)}}$ and ${\sigma'}^2_{r_{\rho\left(i'\right)}} > {\sigma^\star}^2_{r_{\rho\left(i'\right)}}$. Now, notice that $f_i(x) = \log_2 \left( 1 + \frac{\sigma^2_{1_i}\gamma^\star_{s_{\rho\left(i\right)}}}{1+\gamma_{r_i}x} \right)$ is a monotonically decreasing function of $x$. Thus, we have $f_i({\sigma'}^2_{r_{\rho\left(i\right)}}) \leq f_i({\sigma^\star}^2_{r_{\rho\left(i\right)}})$ and $f_{i'}({\sigma'}^2_{r_{\rho\left(i'\right)}}) < f_{i'}({\sigma^\star}^2_{r_{\rho\left(i'\right)}})$.
Adding all above inequalities, we get
\begin{align}
\sum_{i=1}^{\min(M,K_{r})} \log_2{\left(1+
\frac{\sigma_{1_i}^2\gamma^\star_{{s}_{\rho\left(i\right)}}}{1+\gamma_{{r}_i}{\sigma'}^2_{{r}_{\rho{\left(i\right)}}}}\right)}& <\\ \sum_{i=1}^{\min(M,K_{r})} &\log_2{\left(1+
\frac{\sigma_{1_i}^2\gamma^\star_{{s}_{\rho\left(i\right)}}}{1+\gamma_{{r}_i}{\sigma^\star}^2_{{r}_{\rho\left(i\right)}}}\right)}.
\end{align}
The above equation indicates $R^{\text{FD}}_{\text{sr}}(\boldsymbol{\sigma^\star}^2_{r})>R^{\text{FD}}_{\text{sr}}(\boldsymbol{\sigma}^2_{r})$ which contradicts the first assumption $R^{\text{FD}}_{\text{sr}}(\boldsymbol{\sigma^\star}^2_{r})\leq R^{\text{FD}}_{\text{sr}}(\boldsymbol{\sigma}^2_{r})$. Then, the proof of the minimization part is complete.

For the maximization part, the general idea is the same. Again, the proof is by contradiction. We assume the optimal vector $\boldsymbol{\gamma}^\star_{{s}}$, for which we have $R^{\text{FD}}_{\text{sr}}(\boldsymbol{\gamma}^\star_{{s}}) \geq R^{\text{FD}}_{\text{sr}}(\boldsymbol{\gamma}_{{s}})$, does not sum to $P_{s}$. Therefore, we have $\|\boldsymbol{\gamma}_{{s}}\|_1 < P_{s}$. Then there exists $\varepsilon > 0$ for which we have $\|\boldsymbol{\gamma}_{{s}}\|_1 +\varepsilon = P_{s}$. Now we define 
\begin{align}
    \varepsilon_i = \frac{\varepsilon}{\eta}\left( \frac{1+{\sigma^\star}^2_{r_i}\gamma_{r_i}}{\sigma^2_{1_i}}+\gamma^\star_{s_i}\right),
\end{align}
where, $\eta = \sum_i \left(\frac{1+{\sigma^\star}^2_{r_i}\gamma_{r_i}}{\sigma^2_{1_i}}+\gamma^\star_{s_i}\right)$. Now we define the new source power allocation as below
\begin{align}
    {\gamma'}_{s_{\rho\left(i\right)}}= \gamma^\star_{s_{\rho\left(i\right)}}+\varepsilon_i.
\end{align}
One can check that $\sum_i {\gamma'}_{s_{\rho\left(i\right)}} = P_s$ and $\sigma^2_{1_i}{\gamma'}_{s_{\rho(i)}} \geq \sigma^2_{1_{i+1}}{\gamma'}_{s_{\rho(i+1)}}$. Thus, the new source power allocation is in feasible set. Now the remaining is to make sure the new allocation does not change the corresponding $\boldsymbol{\sigma}^2_\mathrm{r}$. Using Lagrangian multiplier we have
\begin{align}
L &= \sum_{i} \!\log_2{\!\left(1+
\frac{\sigma_{1_i}^2{\gamma'}_{{s}_{\rho{\left(i\right)}}}}{1+\gamma_{{r}_i}\sigma_{{s}_{\rho{\left(i\right)}}}^2}\right)} + \lambda \left( \sum_{i=0}^{N}\sigma^2_{r_i} - T\right),\\
&= \sum_{i} \!\log_2{\!\left(1+
\frac{\sigma_{1_i}^2(\gamma^\star_{{s}_{\rho{\left(i\right)}}}+\varepsilon_i)}{1+\gamma_{{r}_i}\sigma_{{s}_{\rho{\left(i\right)}}}^2}\right)} + \lambda \left( \sum_{i=0}^{N}\sigma^2_{r_i} - T\right),\\
&= \sum_{i} \!\log_2{\!\left((1+\frac{\varepsilon}{\eta})\left(1+
\frac{\sigma_{1_i}^2\gamma^\star_{{s}_{\rho{\left(i\right)}}}}{1+\gamma_{{r}_i}\sigma_{{s}_{\rho{\left(i\right)}}}^2}\right)\right)} \nonumber \\ 
&\ \ \ \ + \lambda \left( \sum_{i=0}^{N}\sigma^2_{r_i} - T\right),\\
&= \sum_{i} \!\log_2{\!(1+\frac{\varepsilon}{\eta})}+\sum_i\log_2{\left(1+
\frac{\sigma_{1_i}^2\gamma^\star_{{s}_{\rho{\left(i\right)}}}}{1+\gamma_{{r}_i}\sigma_{{s}_{\rho{\left(i\right)}}}^2}\right)} \nonumber \\ 
&\ \ \ \ + \lambda \left( \sum_{i=0}^{N}\sigma^2_{r_i} - T\right).\\
\end{align}
Now notice that we have $\frac{\partial \sum_{i} \!\log_2{\!(1+\frac{\varepsilon}{\eta})}}{\partial \sigma^2_{\mathrm{r}_i}}=0$ and $\frac{\partial \sum_{i} \!\log_2{\!(1+\frac{\varepsilon}{\eta})}}{\partial \lambda}=0$. As a result, the optimum interference allocation for $\boldsymbol{\gamma}'_\mathrm{r}$
 is the same as that of $\boldsymbol{\gamma}^\star_\mathrm{r}$. Similarly to the case of minimization, here we have $\sum_{i}\varepsilon_i= \varepsilon$. Also we have $\varepsilon_i \geq 0$ and there exist at least one $i'$ for which we have $\varepsilon_{i'} > 0$. Finally as $f_i(x) = \log \left(1+\frac{\sigma^2_{1_i}x}{1+{\sigma^\star}^2_{r_{\rho\left(i\right)}}\gamma_{r_i}}\right)$ is a monotonically increasing function of $x$, we conclude $R^{\text{FD}}_{\text{sr}}(\boldsymbol{\gamma}^\star_{s})<R^{\text{FD}}_{\text{sr}}(\boldsymbol{{\gamma'}}_{s})$ which contradicts the first assumption of $\boldsymbol{\gamma}^\star_{s}$ being the optimal source power allocation, and the proof is complete.
\section*{Appendix III}
First, we show $R^{\text{FD}}_\text{sr}$ is a decreasing function of $T$ and an increasing function of ${P}_{s}$. It is sufficient to show $\frac{d R^{\text{FD}}_\text{sr}}{d P_{{s}}} \geq 0$ and $\frac{d R^{\text{FD}}_\text{sr}}{d T} \leq 0$. We have
\begin{align}
        \frac{d R^{\text{FD}}_\text{sr}}{d P_\text{s}} &=\frac{\sum_{i}^{} \frac{\partial R^{\text{FD}}_\text{sr}}{\partial \gamma_{s_{\rho\left(i\right)}}}d \gamma_{s_{\rho\left(i\right)}}}{\sum_{i}^{} \frac{\partial P_\text{s}}{\partial \gamma_{s_{\rho\left(i\right)}}}d \gamma_{s_{\rho\left(i\right)}}} = \frac{\sum_{i}^{} \frac{\sigma^2_{1_i}}{1+\sigma^2_{r_{\rho\left(i\right)}}\gamma_{r_i}+\sigma^2_{1_i}\gamma_{s_{\rho\left(i\right)}}}d \gamma_{s_{\rho\left(i\right)}}}{\sum_{i}^{} d \gamma_{s_{\rho\left(i\right)}}}\\& \geq \frac{\sum_{i}^{}\phi_1d \gamma_{s_{\rho\left(i\right)}}}{\sum_{i}^{} d \gamma_{s_{\rho\left(i\right)}}}=\phi_1>0,\\
    \frac{d R^{\text{FD}}_\text{sr}}{d T} &=\frac{\sum_{i}^{} \frac{\partial R^{\text{FD}}_\text{sr}}{\partial \sigma^2_{r_{\rho\left(i\right)}}}d \sigma^2_{r_{\rho\left(i\right)}}}{\sum_{i}^{} \frac{\partial T}{\partial \sigma^2_{r_{\rho\left(i\right)}}}d \sigma^2_{r_{\rho\left(i\right)}}} \\&= \frac{\sum_{i}^{}\frac{-\sigma^2_{1_{i}}\gamma_{s_{\rho\left(i\right)}}\gamma_{\mathrm{r}_i}}{\left(1+\sigma^2_{r_{\rho\left(i\right)}}\gamma_{r_i}\right)\left(1+\sigma^2_{1_{i}}\gamma_{s_{\rho\left(i\right)}}+\sigma^2_{r_{i}}\gamma_{r_i}\right)}d \sigma^2_{r_{\rho\left(i\right)}}}{\sum_{i}^{} d \sigma^2_{r_{\rho\left(i\right)}}}\\& \leq \frac{\sum_{i}^{}-\phi_2d \sigma^2_{r_{\rho\left(i\right)}}}{\sum_{i}^{} d \sigma^2_{r_{\rho\left(i\right)}}}=-\phi_2\leq0,
\end{align}
where 
\begin{align}
    \phi_1\stackrel{.}{=} \min_i \Bigg\{ \frac{\sigma^2_{1_i}}{1+\sigma^2_{r_{\rho\left(i\right)}}\gamma_{r_i}+\sigma^2_{1_i}\gamma_{s_{\rho\left(i\right)}}}\Bigg\}
\end{align} 
and 
\begin{align}
    \phi_2\stackrel{.}{=} \min_i \Bigg\{\frac{-\sigma^2_{1_{i}}\gamma_{s_{\rho\left(i\right)}}\gamma_{r_i}}{\left(1+\sigma^2_{r_{\rho\left(i\right)}}\gamma_{r_i}\right)\left(1+\sigma^2_{1_{i}}\gamma_{s_{\rho\left(i\right)}}+\sigma^2_{r_{\rho\left(i\right)}}\gamma_{r_i}\right)}\Bigg\}
\end{align}
respectively.

Next, we show $g(P_{r})=R^{\text{FD}}_\text{sr}(P_{r})-R^{\text{FD}}_\text{rd}(P_{r})$ is a monotonically decreasing function of ${P}_{r}$. It is sufficient to show $\frac{d R^{\text{FD}}_\text{sr}}{d P_{{r}}} \leq 0$ and $\frac{d R^{\text{FD}}_\text{rd}}{d P_{{r}}} > 0$. We have:
\begin{align}
    d R^{\text{FD}}_\text{rd} = \sum_{i}^{}& \frac{\partial R^{\text{FD}}_\text{rd}}{\partial \gamma_{r_i}}d \gamma_{r_i}= \sum_{i}^{} \frac{\sigma^2_{2_i}}{1+\sigma^2_{2_i} \gamma_{r_i}}d \gamma_{r_i}\\
    d R^{\text{FD}}_\text{sr} = \sum_{i}^{}& \frac{\partial R^{\text{FD}}_\text{sr}}{\partial \gamma_{r_i}}d \gamma_{r_i}=\\
    &\sum_{i}^{} \frac{-\sigma^2_{1_{i}}\gamma_{\mathrm{s}_{\rho\left(i\right)}}\sigma^2_{\mathrm{r}_{\rho\left(i\right)}}}{\left(1+\sigma^2_{\mathrm{r}_{\rho\left(i\right)}}\gamma_{\mathrm{r}_i}\right)\left(1+\sigma^2_{1_{i}}\gamma_{\mathrm{s}_{\rho\left(i\right)}}+\sigma^2_{\mathrm{r}_{\rho\left(i\right)}}\gamma_{\mathrm{r}_i}\right)}d \gamma_{\mathrm{r}_i}\\
    d P_\text{r} = \sum_{i}^{}& \frac{\partial P_\text{r}}{\partial \gamma_{\mathrm{r}_i}}d \gamma_{\mathrm{r}_i}=\sum_{i}^{} d \gamma_{\mathrm{r}_i}.
\end{align}
Now we define
\begin{align}
     &\psi_1 \stackrel{.}{=} \min_i \Bigg\{ \frac{\sigma^2_{2_i}}{1+\sigma^2_{2_i}\gamma_{r_i}}\Bigg\}\\
     &\psi_2 \stackrel{.}{=} \min_i \Bigg\{ \frac{\sigma^2_{1_{i}}\gamma_{s_{\rho\left(i\right)}}\sigma^2_{r_{\rho\left(i\right)}}}{\left(1+\sigma^2_{r_{\rho\left(i\right)}}\gamma_{r_i}\right)\left(1+\sigma^2_{1_{i}}\gamma_{s_{\rho\left(i\right)}}+\sigma^2_{r_{\rho\left(i\right)}}\gamma_{r_i}\right)}\Bigg\}.
\end{align}
It is obvious that $\psi_1>0$ and $\psi_2\geq0$. Now we have
\begin{align}
    \frac{d R^{\text{FD}}_\text{sr}}{d P_{{r}}} &= \frac{\sum_{i}^{} \frac{-\sigma^2_{1_{i}}\gamma_{s_{\rho\left(i\right)}}\sigma^2_{r_{\rho\left(i\right)}}}{\left(1+\sigma^2_{r_{\rho\left(i\right)}}\gamma_{r_i}\right)\left(1+\sigma^2_{1_{i}}\gamma_{s_{\rho\left(i\right)}}+\sigma^2_{r_{\rho\left(i\right)}}\gamma_{r_i}\right)}d \gamma_{r_i}}{\sum_{i}^{}d \gamma_{r_i}}\\
    &\leq \frac{\sum_{i}^{} -\psi_2 d \gamma_{r_i}}{\sum_{i}^{}d \gamma_{r_i}} = -\psi_2 \leq 0,
\end{align}
and 
\begin{align}
    \frac{d R^{\text{FD}}_\text{rd}}{d P_{{r}}} &= \frac{\sum_{i}^{} \frac{\sigma^2_{2_i}}{1+\sigma^2_{2_i}\gamma_{r_i}}d \gamma_{r_i}}{\sum_{i}^{}d \gamma_{r_i}}\\
    &\geq \frac{\sum_{i}^{} \psi_1 d \gamma_{r_i}}{\sum_{i}^{}d \gamma_{r_i}} = \psi_1 > 0.
\end{align}
Finally, one can easily conclude
\begin{align}
    \frac{d g}{d P_{\text{r}}} &= \frac{d R^{\text{FD}}_\text{sr}}{d P_{{r}}}-\frac{d R^{\text{FD}}_\text{rd}}{d P_{{r}}}\\
    &\leq -\psi_2 - \psi_1\\
    & <0.
\end{align}
\section*{Appendix IV}
Here we show that if $\gamma_{r_{i}} \geq \gamma_{r_{i+1}}$ and $\sigma_{1_i}^2\gamma_{{s}_{\rho{\left(i\right)}}} \geq \sigma_{1_{i+1}}^2\gamma_{{s}_{\rho\left(i+1\right)}}$ then: $\gamma_{{r}_i}\sigma_{{s}_{\rho{\left(i\right)}}}^2 \geq \gamma_{{r}_{i+1}}\sigma_{{r}_{\rho\left(i+1\right)}}^2$. First we define $f(x,y) = \sqrt{x^2 + axy}-x-b, \ x \geq 0, \ y \geq 0$ in which $a$ and $b$ are positive constants. Now we have,
\begin{align}
    \frac{\partial f}{\partial y} = \frac{ax}{2\sqrt{x^2+axy}} \geq 0.
\end{align}
Also for $\frac{\partial f}{\partial x}$ we have
\begin{align}
    \frac{\partial f}{\partial x} = \frac{2x+ay}{2\sqrt{x^2+axy}} - 1\geq0.
\end{align}
One can check that for positive values $x, \ y$ and $a$, we always have $\frac{2x+ay}{2\sqrt{x^2+axy}} \geq 1$. As a result $f$ is an increasing function of both $x$ and $y$. The rest of the proof is straightforward. We have
\begin{align}
    &\sigma^2_{r_i}\gamma_{r_i} = \left[\frac{\sqrt{\left(\sigma^2_{1_i}\gamma_{s_{\rho\left(i\right)}}\right)^2+\frac{4\sigma^2_{1_i}\gamma_{s_{\rho\left(i\right)}}\gamma_{r_i}}{\lambda}}-\sigma^2_{1_i}\gamma_{s_{\rho\left(i\right)}}-2}{2}\right]^+ \stackrel{(a)}{\geq}\\
    &\left[\frac{\sqrt{\left(\sigma^2_{1_i}\gamma_{s_{\rho\left(i\right)}}\right)^2+\frac{4\sigma^2_{1_i}\gamma_{s_{\rho\left(i\right)}}\gamma_{r_{i+1}}}{\lambda}}-\sigma^2_{1_i}\gamma_{s_{\rho\left(i\right)}}-2}{2}\right]^+ \stackrel{(b)}{\geq}\\
    &\left[\frac{\sqrt{\left(\sigma^2_{1_{i+1}}\gamma_{s_{\rho\left(i+1\right)}}\right)^2\!+\!\frac{4\sigma^2_{1_{i+1}}\gamma_{s_{\rho\left(i+1\right)}}\gamma_{r_{i+1}}}{\lambda}}\!-\!\sigma^2_{1_{i+1}}\gamma_{s_{\rho\left(i+1\right)}}\!-\!2}{2}\right]^+\\
    &=\sigma^2_{r_{i+1}}\gamma_{r_{i+1}},
\end{align}
in which $(a)$ holds because $\gamma_{r_{i}} \geq \gamma_{r_{i+1}}$ and $(b)$ holds because $\sigma_{1_i}^2\gamma_{s_{\rho{\left(i\right)}}} \geq \sigma_{1_{i+1}}^2\gamma_{s_{\rho\left(i+1\right)}}$.
\color{black}
\bibliographystyle{IEEEtran} 
\bibliography{reference}
\end{document}

%% file: SystemModel.tex
\begin{tikzpicture}
\draw (0,0) rectangle (1,2.5);
\TxAntenna{1}{1.8}{0.8};
\draw (1.2,1.6) circle (0.01cm);
\draw (1.2,1.4) circle (0.01cm);
\draw (1.2,1.2) circle (0.01cm);
\TxAntenna{1}{0.2}{0.8};
\node at (1.3,1.8){1};
\node at (1.3,0.2){$M$};

\node[rotate=90] at (0.5,1.25){Source};

\draw[fill=green,opacity=0.5] (3,0) rectangle (5,2.5);

\RxAntenna{3}{1.8}{0.8};
\node at (2.7,1.8){1};
\node at (2.7,0.2){$K_r$};
\draw (2.8,1.6) circle (0.01cm);
\draw (2.8,1.4) circle (0.01cm);
\draw (2.8,1.2) circle (0.01cm);
\RxAntenna{3}{0.2}{0.8};

\TxAntenna{5}{1.8}{0.8};
\node at (5.3,1.8){1};
\node at (5.3,0.2){$K_t$};
\draw (5.2,1.6) circle (0.01cm);
\draw (5.2,1.4) circle (0.01cm);
\draw (5.2,1.2) circle (0.01cm);
\TxAntenna{5}{0.2}{0.8};

\node at (4,1.25){Relay};

\draw (7,0) rectangle (8,2.5);
\RxAntenna{7}{1.8}{0.8};
\node at (6.7,1.8){1};
\node at (6.7,0.2){$N$};
\draw (6.8,1.6) circle (0.01cm);
\draw (6.8,1.4) circle (0.01cm);
\draw (6.8,1.2) circle (0.01cm);
\RxAntenna{7}{0.2}{0.8};

\node[rotate=90] at (7.5,1.25){Destination};

\draw[fill=yellow,opacity=0.3] (2,1.25) ellipse (0.5cm and 1.25cm);
\node at (2,1.25){$\mathbf{H}_1$};
\draw[fill=yellow,opacity=0.3] (6,1.25) ellipse (0.5cm and 1.25cm);
\node at (6,1.25){$\mathbf{H}_2$};

\draw[thick,->] (5.2,2.7) edge[out=90, in=90] (2.8,2.7);
\node at (4,3.7){$\mathbf{H}_\mathrm{r}$};

\end{tikzpicture}

%% file: WaterLevel.tex
\begin{tikzpicture}
\draw[->] (0,0)--(5,0);
\draw[->] (0,0)--(0,6);
\node at (5.2,0){$i$};
\node at (0,6.4){$\frac{\sigma_{1_i}^2}{1+\gamma_{\mathrm{r}_i}\sigma_{\mathrm{r}_i}^2}$};
\draw[fill=brown,opacity=0.8] (0,0) rectangle (.25,2);
\draw[fill=brown,opacity=0.8] (.25,0) rectangle (.5,4);
\draw[fill=brown,opacity=0.8] (.5,0) rectangle (.75,3.8);
\draw[fill=brown,opacity=0.8] (.75,0) rectangle (1,3.6);
\draw[fill=brown,opacity=0.8] (1,0) rectangle (1.25,3.3);
\draw[fill=brown,opacity=0.8] (1.25,0) rectangle (1.5,2.5);
\draw[fill=brown,opacity=0.8] (1.5,0) rectangle (1.75,2.5);
\draw[fill=brown,opacity=0.8] (1.75,0) rectangle (2,2.2);
\draw[fill=brown,opacity=0.8] (2,0) rectangle (2.25,2.3);
\draw[fill=brown,opacity=0.8] (2.25,0) rectangle (2.5,2.4);
\draw[fill=brown,opacity=0.8] (2.5,0) rectangle (2.75,2.8);
\draw[fill=brown,opacity=0.8] (2.75,0) rectangle (3,3.1);
\draw[fill=brown,opacity=0.8] (3,0) rectangle (3.25,3.4);
\draw[fill=brown,opacity=0.8] (3.25,0) rectangle (3.5,2.8);
\draw[fill=brown,opacity=0.8] (3.5,0) rectangle (3.75,2.6);
\draw[fill=brown,opacity=0.8] (3.75,0) rectangle (4,2.2);
\fill[fill=blue] (0,2) rectangle (.26,4.7);
\fill[fill=blue] (.25,4) rectangle (.51,4.7);
\fill[fill=blue] (.5,3.8) rectangle (.76,4.7);
\fill[fill=blue] (.75,3.6) rectangle (1.01,4.7);
\fill[fill=blue] (1,3.3) rectangle (1.26,4.5);
\fill[fill=blue] (1.25,2.5) rectangle (1.51,4.7);
\fill[fill=blue] (1.5,2.5) rectangle (1.76,4.7);
\fill[fill=blue] (1.75,2.2) rectangle (2.01,4.7);
\fill[fill=blue] (2,2.3) rectangle (2.26,4.6);
\fill[fill=blue] (2.25,2.4) rectangle (2.51,4.5);
\fill[fill=blue] (2.5,2.8) rectangle (2.76,4.3);
\fill[fill=blue] (2.74,3.1) rectangle (3.01,4.2);
\fill[fill=blue] (3,3.4) rectangle (3.26,3.7);
\fill[fill=blue] (3.25,2.8) rectangle (3.51,4.1);
\fill[fill=blue] (3.5,2.6) rectangle (3.76,4);
\fill[fill=blue] (3.75,2.2) rectangle (4,3.2);

\draw[pattern=crosshatch,pattern color=gray!80!black][fill=gray,opacity=0.8] (.25,5) rectangle (.5,5.5);
\draw[pattern=crosshatch,pattern color=gray!80!black][fill=gray,opacity=0.8] (.5,4.9) rectangle (.75,5.5);
\draw[pattern=crosshatch,pattern color=gray!80!black][fill=gray,opacity=0.8] (.75,4.8) rectangle (1,5.5);
\draw[pattern=crosshatch,pattern color=gray!80!black][fill=gray,opacity=0.8] (1,4.5) rectangle (1.25,5.5);
\draw[pattern=crosshatch,pattern color=gray!80!black][fill=gray,opacity=0.8] (1.25,4.8) rectangle (1.5,5.5);
\draw[pattern=crosshatch,pattern color=gray!80!black][fill=gray,opacity=0.8] (1.5,4.8) rectangle (1.75,5.5);
\draw[pattern=crosshatch,pattern color=gray!80!black][fill=gray,opacity=0.8] (1.75,4.7) rectangle (2,5.5);
\draw[pattern=crosshatch,pattern color=gray!80!black][fill=gray,opacity=0.8] (2,4.6) rectangle (2.25,5.5);
\draw[pattern=crosshatch,pattern color=gray!80!black][fill=gray,opacity=0.8] (2.25,4.5) rectangle (2.5,5.5);
\draw[pattern=crosshatch,pattern color=gray!80!black][fill=gray,opacity=0.8] (2.5,4.3) rectangle (2.75,5.5);
\draw[pattern=crosshatch,pattern color=gray!80!black][fill=gray,opacity=0.8] (2.75,4.2) rectangle (3,5.5);
\draw[pattern=crosshatch,pattern color=gray!80!black][fill=gray,opacity=0.8] (3,3.7) rectangle (3.25,5.5);
\draw[pattern=crosshatch,pattern color=gray!80!black][fill=gray,opacity=0.8] (3.25,4.1) rectangle (3.5,5.5);
\draw[pattern=crosshatch,pattern color=gray!80!black][fill=gray,opacity=0.8] (3.5,4) rectangle (3.75,5.5);
\draw[pattern=crosshatch,pattern color=gray!80!black][fill=gray,opacity=0.8] (3.75,3.2) rectangle (4,5.5);
\end{tikzpicture}

%% file: WaterLevelB.tex
\begin{tikzpicture}
\draw[->] (0,0)--(5,0);
\draw[->] (0,0)--(0,6);
\node at (5.2,0){$i$};
\node at (0,6.4){$\frac{\sigma_{1_i}^2}{1+\gamma_{\mathrm{r}_i}\sigma_{\mathrm{r}_i}^2}$};
\draw[fill=brown,opacity=0.8] (0,0) rectangle (.25,2);
\draw[fill=brown,opacity=0.8] (.25,0) rectangle (.5,4);
\draw[fill=brown,opacity=0.8] (.5,0) rectangle (.75,3.8);
\draw[fill=brown,opacity=0.8] (.75,0) rectangle (1,3.6);
\draw[fill=brown,opacity=0.8] (1,0) rectangle (1.25,3.3);
\draw[fill=brown,opacity=0.8] (1.25,0) rectangle (1.5,2.5);
\draw[fill=brown,opacity=0.8] (1.5,0) rectangle (1.75,2.5);
\draw[fill=brown,opacity=0.8] (1.75,0) rectangle (2,2.2);
\draw[fill=brown,opacity=0.8] (2,0) rectangle (2.25,2.3);
\draw[fill=brown,opacity=0.8] (2.25,0) rectangle (2.5,2.4);
\draw[fill=brown,opacity=0.8] (2.5,0) rectangle (2.75,2.8);
\draw[fill=brown,opacity=0.8] (2.75,0) rectangle (3,3.1);
\draw[fill=brown,opacity=0.8] (3,0) rectangle (3.25,3.7);
\draw[fill=brown,opacity=0.8] (3.25,0) rectangle (3.5,2.8);
\draw[fill=brown,opacity=0.8] (3.5,0) rectangle (3.75,2.6);
\draw[fill=brown,opacity=0.8] (3.75,0) rectangle (4,2.2);
\fill[fill=blue] (0,2) rectangle (.26,4.7);
\fill[fill=blue] (.25,4) rectangle (.51,4.7);
\fill[fill=blue] (.5,3.8) rectangle (.76,4.7);
\fill[fill=blue] (.75,3.6) rectangle (1.01,4.7);
\fill[fill=blue] (1,3.3) rectangle (1.26,4.5);
\fill[fill=blue] (1.25,2.5) rectangle (1.51,4.7);
\fill[fill=blue] (1.5,2.5) rectangle (1.76,4.7);
\fill[fill=blue] (1.75,2.2) rectangle (2.01,4.7);
\fill[fill=blue] (2,2.3) rectangle (2.26,4.6);
\fill[fill=blue] (2.25,2.4) rectangle (2.51,4.5);
\fill[fill=blue] (2.5,2.8) rectangle (2.76,4.3);
\fill[fill=blue](2.75,3.1) rectangle (3,4.2);

\draw[pattern=crosshatch,pattern color=gray!60!black][fill=gray,opacity=0.8](.25,5) rectangle (.5,5.5);
\draw[pattern=crosshatch,pattern color=gray!80!black][fill=gray,opacity=0.8] (.5,4.9) rectangle (.75,5.5);
\draw[pattern=crosshatch,pattern color=gray!80!black][fill=gray,opacity=0.8] (.75,4.8) rectangle (1,5.5);
\draw[pattern=crosshatch,pattern color=gray!80!black][fill=gray,opacity=0.8] (1,4.5) rectangle (1.25,5.5);
\draw[pattern=crosshatch,pattern color=gray!80!black][fill=gray,opacity=0.8] (1.25,4.8) rectangle (1.5,5.5);
\draw[pattern=crosshatch,pattern color=gray!80!black][fill=gray,opacity=0.8] (1.5,4.8) rectangle (1.75,5.5);
\draw[pattern=crosshatch,pattern color=gray!80!black][fill=gray,opacity=0.8] (1.75,4.7) rectangle (2,5.5);
\draw[pattern=crosshatch,pattern color=gray!80!black][fill=gray,opacity=0.8] (2,4.6) rectangle (2.25,5.5);
\draw[pattern=crosshatch,pattern color=gray!80!black][fill=gray,opacity=0.8] (2.25,4.5) rectangle (2.5,5.5);
\draw[pattern=crosshatch,pattern color=gray!80!black][fill=gray,opacity=0.8] (2.5,4.3) rectangle (2.75,5.5);
\draw[pattern=crosshatch,pattern color=gray!80!black][fill=gray,opacity=0.8] (2.75,4.2) rectangle (3,5.5);
\draw[pattern=crosshatch,pattern color=gray!80!black][fill=gray,opacity=0.8] (3,3.7) rectangle (3.25,5.5);
\draw[pattern=crosshatch,pattern color=gray!80!black][fill=gray,opacity=0.8] (3.25,2.8) rectangle (3.5,5.5);
\draw[pattern=crosshatch,pattern color=gray!80!black][fill=gray,opacity=0.8] (3.5,2.6) rectangle (3.75,5.5);
\draw[pattern=crosshatch,pattern color=gray!80!black][fill=gray,opacity=0.8] (3.75,2.2) rectangle (4,5.5);
\end{tikzpicture}

%% file: myfile.tex
%
%

\definecolor{mycolor1}{rgb}{0.00000,0.44700,0.74100}%
\definecolor{mycolor2}{rgb}{0.85000,0.32500,0.09800}%
\definecolor{mycolor3}{rgb}{0.92900,0.69400,0.12500}%
\definecolor{mycolor4}{rgb}{0.49400,0.18400,0.55600}%
\definecolor{mycolor5}{rgb}{0.46600,0.67400,0.18800}%
\definecolor{mycolor6}{rgb}{0.30100,0.74500,0.93300}%
\definecolor{mycolor7}{rgb}{0.63500,0.07800,0.18400}%
\begin{tikzpicture}
\begin{axis}[%
width=4.521in,
height=3.566in,
at={(0.758in,0.481in)},
scale only axis,
xmin=0,
xmax=60,
xlabel style={font=\color{white!15!black}},
xlabel near ticks,
xlabel={$\frac{T}{P}$},
ymin=3,
ymax=13,
ylabel style={font=\color{white!15!black}},
ylabel near ticks,
ylabel={$R_\mathrm{av}$ (bits/channel use)},
axis background/.style={fill=white},
xmajorgrids,
ymajorgrids,
legend style={at={(axis cs: 60,13)}, anchor=north east,draw=black,fill=white, fill opacity=.8,legend cell align=left}
]
\addplot [color=mycolor1, mark=+, mark options={solid, mycolor1}]
  table[row sep=crcr]{%
0	10.0132066674453\\
3	7.04082801139137\\
6	6.1654012429215\\
9	5.64903112438047\\
12	5.28902178187085\\
15	5.01579611800772\\
18	4.79793569048092\\
21	4.61871236239081\\
24	4.46717281167577\\
27	4.33687706025406\\
30	4.22321618534349\\
33	4.12291286377084\\
36	4.03386458815778\\
39	3.95296533749603\\
42	3.88068920746087\\
45	3.81439871455062\\
48	3.7544434109944\\
51	3.6980729048974\\
54	3.64675795929692\\
57	3.59937575669153\\
60	3.55503215518893\\
};
\addlegendentry{$\text{FD=}{\text{4,3,7,4}}$}

\addplot [color=mycolor2, dashed, mark=+, mark options={solid, mycolor2}]
   table[row sep=crcr]{%
0	6.00387678867345\\
3	6.00387678867345\\
6	6.00387678867345\\
9	6.00387678867345\\
12	6.00387678867345\\
15	6.00387678867345\\
18	6.00387678867345\\
21	6.00387678867345\\
24	6.00387678867345\\
27	6.00387678867345\\
30	6.00387678867345\\
33	6.00387678867345\\
36	6.00387678867345\\
39	6.00387678867345\\
42	6.00387678867345\\
45	6.00387678867345\\
48	6.00387678867345\\
51	6.00387678867345\\
54	6.00387678867345\\
57	6.00387678867345\\
60	6.00387678867345\\
};
\addlegendentry{$\text{HD=}{\text{4,3,7,4}}$}

\addplot [color=mycolor3, mark=o, mark options={solid, mycolor3}]
  table[row sep=crcr]{%
0	11.4993255963629\\
3	7.7395785898021\\
6	6.79592903828924\\
9	6.25328510393818\\
12	5.88026726335178\\
15	5.59985124042908\\
18	5.37742097629068\\
21	5.19537081762894\\
24	5.04248546314715\\
27	4.91156263075019\\
30	4.79753125595011\\
33	4.69739678881863\\
36	4.60838009947396\\
39	4.52837566850092\\
42	4.45611185815136\\
45	4.3904419608254\\
48	4.33057293738069\\
51	4.27561271116077\\
54	4.22518264470804\\
57	4.17841701941906\\
60	4.13470442796042\\
};
\addlegendentry{$\text{FD=}{\text{4,4,6,4}}$}

\addplot [color=mycolor4, dashed, mark=o, mark options={solid, mycolor4}]
  table[row sep=crcr]{%
0	6.30105258604983\\
3	6.30105258604983\\
6	6.30105258604983\\
9	6.30105258604983\\
12	6.30105258604983\\
15	6.30105258604983\\
18	6.30105258604983\\
21	6.30105258604983\\
24	6.30105258604983\\
27	6.30105258604983\\
30	6.30105258604983\\
33	6.30105258604983\\
36	6.30105258604983\\
39	6.30105258604983\\
42	6.30105258604983\\
45	6.30105258604983\\
48	6.30105258604983\\
51	6.30105258604983\\
54	6.30105258604983\\
57	6.30105258604983\\
60	6.30105258604983\\
};
\addlegendentry{$\text{HD=}{\text{4,4,6,4}}$}

\addplot [color=mycolor5, mark=asterisk, mark options={solid, mycolor5}]
  table[row sep=crcr]{%
0	12.2068913276888\\
3	8.30927811834808\\
6	7.37664452595049\\
9	6.84175417166565\\
12	6.47337729829713\\
15	6.19610476088662\\
18	5.97665050383111\\
21	5.79624479122752\\
24	5.64435375471356\\
27	5.51412934445539\\
30	5.40062799771128\\
33	5.30055270991646\\
36	5.21148200470566\\
39	5.13082020193476\\
42	5.05908689969377\\
45	4.99386992817086\\
48	4.93402389539288\\
51	4.8787135109579\\
54	4.82741700509176\\
57	4.78010305883844\\
60	4.73547708244634\\
};
\addlegendentry{$\text{FD=}{\text{4,5,5,4}}$}

\addplot [color=mycolor6, dashed, mark=asterisk, mark options={solid, mycolor6}]
  table[row sep=crcr]{%
0	6.41041779331438\\
3	6.41041779331438\\
6	6.41041779331438\\
9	6.41041779331438\\
12	6.41041779331438\\
15	6.41041779331438\\
18	6.41041779331438\\
21	6.41041779331438\\
24	6.41041779331438\\
27	6.41041779331438\\
30	6.41041779331438\\
33	6.41041779331438\\
36	6.41041779331438\\
39	6.41041779331438\\
42	6.41041779331438\\
45	6.41041779331438\\
48	6.41041779331438\\
51	6.41041779331438\\
54	6.41041779331438\\
57	6.41041779331438\\
60	6.41041779331438\\
};
\addlegendentry{$\text{HD=}{\text{4,5,5,4}}$}

\addplot [color=mycolor7, mark=x, mark options={solid, mycolor7}]
  table[row sep=crcr]{%
0	11.4961867674722\\
3	8.58294363923564\\
6	7.74659979088856\\
9	7.27767823514232\\
12	6.9558530387805\\
15	6.71391199922265\\
18	6.52192781946218\\
21	6.36354830552742\\
24	6.23014826129296\\
27	6.11564190628422\\
30	6.01572760738859\\
33	5.92709681546741\\
36	5.84757287955966\\
39	5.77632144110019\\
42	5.71165459888111\\
45	5.65321876701047\\
48	5.59892219896436\\
51	5.54820353226281\\
54	5.50176828719364\\
57	5.45717937505284\\
60	5.41568092394417\\
};
\addlegendentry{$\text{FD=}{\text{4,6,4,4}}$}

\addplot [color=mycolor1, dashed, mark=x, mark options={solid, mycolor1}]
  table[row sep=crcr]{%
0	6.30422022974078\\
3	6.30422022974078\\
6	6.30422022974078\\
9	6.30422022974078\\
12	6.30422022974078\\
15	6.30422022974078\\
18	6.30422022974078\\
21	6.30422022974078\\
24	6.30422022974078\\
27	6.30422022974078\\
30	6.30422022974078\\
33	6.30422022974078\\
36	6.30422022974078\\
39	6.30422022974078\\
42	6.30422022974078\\
45	6.30422022974078\\
48	6.30422022974078\\
51	6.30422022974078\\
54	6.30422022974078\\
57	6.30422022974078\\
60	6.30422022974078\\
};
\addlegendentry{$\text{HD=}{\text{4,6,4,4}}$}

\addplot [color=mycolor2, mark=square, mark options={solid, mycolor2}]
  table[row sep=crcr]{%
0	9.93954145608212\\
3	8.46501364468555\\
6	7.65717126937559\\
9	7.22487394194823\\
12	6.94129369473856\\
15	6.73293136082532\\
18	6.56925364322894\\
21	6.4346665877619\\
24	6.32179438349485\\
27	6.22570709617225\\
30	6.14229123664884\\
33	6.06956016456162\\
36	6.00405164551726\\
39	5.94777588751717\\
42	5.89631612865323\\
45	5.85110797455816\\
48	5.81120441515374\\
51	5.77267135491322\\
54	5.7409776465364\\
57	5.70971962116244\\
60	5.68147002202778\\
};
\addlegendentry{$\text{FD=}{\text{4,7,3,4}}$}

\addplot [color=mycolor3, dashed, mark=square, mark options={solid, mycolor3}]
  table[row sep=crcr]{%
0	5.97231140653149\\
3	5.97230978213356\\
6	5.97230971623485\\
9	5.97230984937951\\
12	5.97230946711029\\
15	5.97230956685565\\
18	5.972309369269\\
21	5.97230943596867\\
24	5.97230930486099\\
27	5.97230931699152\\
30	5.97230927041861\\
33	5.97230924162419\\
36	5.97230925501972\\
39	5.97230919493087\\
42	5.97230921694348\\
45	5.9723091821248\\
48	5.97230917572501\\
51	5.97230917701791\\
54	5.97230914970978\\
57	5.97230916197635\\
60	5.97230913445439\\
};
\addlegendentry{$\text{HD=}{\text{4,7,3,4}}$}
  \coordinate (spypoint) at (axis cs:1.6,7.1);
  \coordinate (magnifyglass) at (axis cs:7.5,9.2);
  \coordinate (spypoint1) at (axis cs:18,5.5);
  \coordinate (magnifyglass1) at (axis cs:22.5,9.2);
  \coordinate (spypoint2) at (axis cs:25.2,5.3);
  \coordinate (magnifyglass2) at (axis cs:28.8,6.7);
\end{axis}

\end{tikzpicture}%

%% file: asli.tex
%
%
\definecolor{mycolor1}{rgb}{0.00000,0.44700,0.74100}%
\definecolor{mycolor2}{rgb}{0.85000,0.32500,0.09800}%
\definecolor{mycolor3}{rgb}{0.92900,0.69400,0.12500}%
\definecolor{mycolor4}{rgb}{0.49400,0.18400,0.55600}%
\definecolor{mycolor5}{rgb}{0.46600,0.67400,0.18800}%
\definecolor{mycolor6}{rgb}{0.30100,0.74500,0.93300}%
\definecolor{mycolor7}{rgb}{0.63500,0.07800,0.18400}%
\begin{tikzpicture}

\begin{axis}[%
width=4.521in,
height=3.566in,
at={(0.758in,0.481in)},
scale only axis,
xmin=0,
xmax=60,
xlabel style={font=\color{white!15!black}},
xlabel near ticks,
xlabel={$\frac{T}{P}$},
ymin=12,
ymax=32,
ylabel style={font=\color{white!15!black}},
ylabel near ticks,
ylabel={$R_\mathrm{av}$ (bits/channel use)},
axis background/.style={fill=white},
xmajorgrids,
ymajorgrids,
legend style={at={(axis cs: 60,32)}, anchor=north east,draw=black,fill=white, fill opacity=0.8,legend cell align=left}
]
\addplot [color=mycolor1, mark=+, mark options={solid, mycolor1}]
  table[row sep=crcr]{%
0	22.2379831813425\\
3	22.23797283972\\
6	22.2342391021216\\
9	22.1184368169953\\
12	21.7501415914825\\
15	21.242222022122\\
18	20.746016654095\\
21	20.3068547475085\\
24	19.9241160653056\\
27	19.5906674199889\\
30	19.2955723136925\\
33	19.033346686762\\
36	18.7976034385335\\
39	18.5833056044367\\
42	18.3868707311458\\
45	18.2056988415776\\
48	18.0371292674408\\
51	17.879372684692\\
54	17.7318767007492\\
57	17.5936642186954\\
60	17.4631869556317\\
};
\addlegendentry{$\text{FD=}{\text{10,18,6,10}}$}

\addplot [color=mycolor2, dashed, mark=+, mark options={solid, mycolor2}]
   table[row sep=crcr]{%
0	14.063407462385\\
3	14.063403315406\\
6	14.0634023526012\\
9	14.0634022045181\\
12	14.0634021630196\\
15	14.0634021548112\\
18	14.0634021864253\\
21	14.0634021478207\\
24	14.0634019752863\\
27	14.0634018567046\\
30	14.0634018176362\\
33	14.063401831277\\
36	14.0634018370433\\
39	14.0634018033364\\
42	14.063401749606\\
45	14.0634017294446\\
48	14.0634017262263\\
51	14.0634017245112\\
54	14.0634017099448\\
57	14.063401685183\\
60	14.063401670426\\
};
\addlegendentry{$\text{HD=}{\text{10,18,6,10}}$}

\addplot [color=mycolor3, mark=o, mark options={solid, mycolor3}]
  table[row sep=crcr]{%
0	28.870246711131\\
3	27.2538178798514\\
6	25.1508574301153\\
9	23.8411680259688\\
12	22.9053052816847\\
15	22.1827360530484\\
18	21.5954392975779\\
21	21.1021269821487\\
24	20.6784430152478\\
27	20.3078328604544\\
30	19.9791838169882\\
33	19.6831838361889\\
36	19.4152378968822\\
39	19.1704278773893\\
42	18.9450854463147\\
45	18.7366628976672\\
48	18.5430350149395\\
51	18.3621404422607\\
54	18.1924955718451\\
57	18.0326300407024\\
60	17.8811891482315\\
};
\addlegendentry{$\text{FD=}{\text{10,14,10,10}}$}

\addplot [color=mycolor4, dashed, mark=o, mark options={solid, mycolor4}]
  table[row sep=crcr]{%
0	15.6093574888593\\
3	15.6093574888593\\
6	15.6093574888593\\
9	15.6093574888593\\
12	15.6093574888593\\
15	15.6093574888593\\
18	15.6093574888593\\
21	15.6093574888593\\
24	15.6093574888593\\
27	15.6093574888593\\
30	15.6093574888593\\
33	15.6093574888593\\
36	15.6093574888593\\
39	15.6093574888593\\
42	15.6093574888593\\
45	15.6093574888593\\
48	15.6093574888593\\
51	15.6093574888593\\
54	15.6093574888593\\
57	15.6093574888593\\
60	15.6093574888593\\
};
\addlegendentry{$\text{HD=}{\text{10,14,10,10}}$}

\addplot [color=mycolor5, mark=asterisk, mark options={solid, mycolor5}]
  table[row sep=crcr]{%
0	30.9353767530573\\
3	26.3177498347916\\
6	24.2066173311158\\
9	22.856935220317\\
12	21.8695883461739\\
15	21.0945930913527\\
18	20.4600393482583\\
21	19.924862246167\\
24	19.4632622923921\\
27	19.0588058840285\\
30	18.6995908219613\\
33	18.3771119889963\\
36	18.0849472258941\\
39	17.8182059991727\\
42	17.5737131813104\\
45	17.3482458395814\\
48	17.1388181678288\\
51	16.9436053853072\\
54	16.7612252013682\\
57	16.5902092414658\\
60	16.4292380129159\\
};
\addlegendentry{$\text{FD=}{\text{10,12,12,10}}$}

\addplot [color=mycolor6, dashed, mark=asterisk, mark options={solid, mycolor6}]
  table[row sep=crcr]{%
0	15.7771979508363\\
3	15.7771979508363\\
6	15.7771979508363\\
9	15.7771979508363\\
12	15.7771979508363\\
15	15.7771979508363\\
18	15.7771979508363\\
21	15.7771979508363\\
24	15.7771979508363\\
27	15.7771979508363\\
30	15.7771979508363\\
33	15.7771979508363\\
36	15.7771979508363\\
39	15.7771979508363\\
42	15.7771979508363\\
45	15.7771979508363\\
48	15.7771979508363\\
51	15.7771979508363\\
54	15.7771979508363\\
57	15.7771979508363\\
60	15.7771979508363\\
};
\addlegendentry{$\text{HD=}{\text{10,12,12,10}}$}

\addplot [color=mycolor7, mark=x, mark options={solid, mycolor7}]
  table[row sep=crcr]{%
  0	28.8059999382025\\
3	24.7451992481848\\
6	22.8273444702011\\
9	21.5564651000405\\
12	20.6087786839763\\
15	19.8563853585685\\
18	19.2348594560757\\
21	18.7066834537158\\
24	18.2490333223378\\
27	17.8466383565836\\
30	17.4878278406791\\
33	17.164819515324\\
36	16.8712239773988\\
39	16.6027776175203\\
42	16.3560753055515\\
45	16.1276169333143\\
48	15.9153256789564\\
51	15.7173378567254\\
54	15.5318741976491\\
57	15.3574488220345\\
60	15.1932387688318\\
};
\addlegendentry{$\text{FD=}{\text{10,10,14,10}}$}

\addplot [color=mycolor1, dashed, mark=x, mark options={solid, mycolor1}]
  table[row sep=crcr]{%
0	15.5895983142936\\
3	15.5895983142936\\
6	15.5895983142936\\
9	15.5895983142936\\
12	15.5895983142936\\
15	15.5895983142936\\
18	15.5895983142936\\
21	15.5895983142936\\
24	15.5895983142936\\
27	15.5895983142936\\
30	15.5895983142936\\
33	15.5895983142936\\
36	15.5895983142936\\
39	15.5895983142936\\
42	15.5895983142936\\
45	15.5895983142936\\
48	15.5895983142936\\
51	15.5895983142936\\
54	15.5895983142936\\
57	15.5895983142936\\
60	15.5895983142936\\
};
\addlegendentry{$\text{HD=}{\text{10,10,14,10}}$}

\addplot [color=mycolor2, mark=square, mark options={solid, mycolor2}]
  table[row sep=crcr]{%
0	22.2308768912884\\
3	20.0185731633967\\
6	18.7098627807661\\
9	17.7729721842391\\
12	17.0431823906632\\
15	16.4463824174531\\
18	15.9422737456717\\
21	15.5067626116879\\
24	15.1233300359948\\
27	14.7811976912979\\
30	14.472735054137\\
33	14.1918739634842\\
36	13.9344530787368\\
39	13.6972262668899\\
42	13.4772719979522\\
45	13.2722811523645\\
48	13.0806401292542\\
51	12.9004621061488\\
54	12.7306011301035\\
57	12.5705098505658\\
60	12.418912841273\\
};
\addlegendentry{$\text{FD=}{\text{10,6,18,10}}$}

\addplot [color=mycolor3, dashed, mark=square, mark options={solid, mycolor3}]
  table[row sep=crcr]{%
0	14.0597126722874\\
3	14.0597126722874\\
6	14.0597126722874\\
9	14.0597126722874\\
12	14.0597126722874\\
15	14.0597126722874\\
18	14.0597126722874\\
21	14.0597126722874\\
24	14.0597126722874\\
27	14.0597126722874\\
30	14.0597126722874\\
33	14.0597126722874\\
36	14.0597126722874\\
39	14.0597126722874\\
42	14.0597126722874\\
45	14.0597126722874\\
48	14.0597126722874\\
51	14.0597126722874\\
54	14.0597126722874\\
57	14.0597126722874\\
60	14.0597126722874\\
};
\addlegendentry{$\text{HD=}{\text{10,6,18,10}}$}

\end{axis}

\end{tikzpicture}%

%% file: CwrtT2.tex
%
%
\definecolor{mycolor1}{rgb}{0.00000,0.44700,0.74100}%
\definecolor{mycolor2}{rgb}{0.85000,0.32500,0.09800}%
\definecolor{mycolor3}{rgb}{0.92900,0.69400,0.12500}%
\definecolor{mycolor4}{rgb}{0.49400,0.18400,0.55600}%
\definecolor{mycolor5}{rgb}{0.46600,0.67400,0.18800}%
\definecolor{mycolor6}{rgb}{0.30100,0.74500,0.93300}%
\begin{tikzpicture}

\begin{axis}[%
width=4.521in,
height=3.566in,
at={(0.758in,0.481in)},
scale only axis,
xmin=1,
xmax=7,
xlabel near ticks,
xlabel={$K_\mathrm{r}$},
ymin=2,
ymax=11,
ylabel near ticks,
ylabel={$R_\mathrm{av}$ (bits/channel use)},
axis background/.style={fill=white},
xmajorgrids,
ymajorgrids,
legend style={at={(axis cs: 7,11)}, anchor=north east,draw=black,fill=white, fill opacity=0.8,legend cell align=left}
]
\addplot [color=mycolor1, mark=square]
  table[row sep=crcr]{%
1	5.19276598801662\\
2	8.11920085133601\\
3	10.027191368564\\
4	10.8515310803683\\
5	9.93848022937452\\
6	8.042436251107\\
7	5.17323523709587\\
};
\addlegendentry{$T=0$}

\addplot [color=mycolor2, mark=+]
  table[row sep=crcr]{%
1	5.17456846237112\\
2	7.91002338412908\\
3	8.70041730192537\\
4	8.44807183953743\\
5	7.76298916526416\\
6	6.5559413410457\\
7	4.40267849337322\\
};
\addlegendentry{$T=5$}

\addplot [color=mycolor3, mark=o]
  table[row sep=crcr]{%
1	5.18233002042795\\
2	7.68005974863779\\
3	7.86960912266989\\
4	7.6299803058935\\
5	6.99447042609095\\
6	5.91885804131884\\
7	4.01999728521179\\
};
\addlegendentry{$T=10$}

\addplot [color=mycolor4, mark=x]
  table[row sep=crcr]{%
1	5.18606979631162\\
2	7.11147904609353\\
3	7.12034585932373\\
4	6.74458850504673\\
5	6.18596301127858\\
6	5.21103269432574\\
7	3.58056935580673\\
};
\addlegendentry{$T=20$}

\addplot [color=mycolor5, mark=diamond]
  table[row sep=crcr]{%
1	5.16649202904227\\
2	6.48664294790432\\
3	6.34659420434169\\
4	5.95948696883438\\
5	5.32331930546157\\
6	4.4781601507804\\
7	3.0444713328328\\
};
\addlegendentry{$T=40$}

\addplot [color=mycolor6, mark=asterisk]
  table[row sep=crcr]{%
1	5.19069523284768\\
2	6.03360938110137\\
3	5.77634142170784\\
4	5.28261490109712\\
5	4.59012563420394\\
6	3.75988412987369\\
7	2.56564375656247\\
};
\addlegendentry{$T=80$}

\end{axis}

\begin{axis}[%
width=5.833in,
height=4.375in,
at={(0in,0in)},
scale only axis,
xmin=0,
xmax=1,
ymin=0,
ymax=1,
axis line style={draw=none},
ticks=none,
axis x line*=bottom,
axis y line*=left,
legend style={legend cell align=left, align=left, draw=white!15!black}
]
\end{axis}
\end{tikzpicture}%

%% file: CwrtT.tex
%
%
\definecolor{mycolor1}{rgb}{0.00000,0.44700,0.74100}%
\definecolor{mycolor2}{rgb}{0.85000,0.32500,0.09800}%
\definecolor{mycolor3}{rgb}{0.92900,0.69400,0.12500}%
\definecolor{mycolor4}{rgb}{0.49400,0.18400,0.55600}%
\definecolor{mycolor5}{rgb}{0.46600,0.67400,0.18800}%
\definecolor{mycolor6}{rgb}{0.30100,0.74500,0.93300}%
\begin{tikzpicture}

\begin{axis}[%
width=4.521in,
height=3.566in,
at={(0.758in,0.481in)},
scale only axis,
xmin=1,
xmax=7,
xlabel near ticks,
xlabel={$K_\mathrm{r}$},
ymin=2,
ymax=16,
ylabel near ticks,
ylabel={$R_\mathrm{av}$ (bits/channel use)},
axis background/.style={fill=white},
xmajorgrids,
ymajorgrids,
legend style={at={(axis cs: 7,16)}, anchor=north east,draw=black,fill=white, fill opacity=0.8,legend cell align=left}
]
\addplot [color=mycolor1, mark=square]
  table[row sep=crcr]{%
1	6.25514647331752\\
2	10.3662453539247\\
3	13.5334750357343\\
4	15.4999872355007\\
5	13.5342850648697\\
6	10.3495384944814\\
7	6.24955898141535\\
};
\addlegendentry{$T=0$}

\addplot [color=mycolor2, mark=asterisk]
  table[row sep=crcr]{%
1	6.27283606544687\\
2	10.3591616140734\\
3	12.9429505789372\\
4	12.9106518932493\\
5	11.4893258490764\\
6	9.13312761863789\\
7	5.68512178872962\\
};
\addlegendentry{$T=5$}

\addplot [color=mycolor3, mark=diamond]
  table[row sep=crcr]{%
1	6.26845092007536\\
2	10.3443158801484\\
3	12.0271630031676\\
4	11.7657338750834\\
5	10.5530918898085\\
6	8.47138420160092\\
7	5.34962459329345\\
};
\addlegendentry{$T=10$}

\addplot [color=mycolor4, mark=x]
  table[row sep=crcr]{%
1	6.24113528053587\\
2	10.3099854735025\\
3	10.8976125021898\\
4	10.5335497110413\\
5	9.46609406910993\\
6	7.72183763043796\\
7	4.87863863026287\\
};
\addlegendentry{$T=20$}

\addplot [color=mycolor5, mark=o]
  table[row sep=crcr]{%
1	6.2591599039995\\
2	10.2304180918633\\
3	9.80037203548379\\
4	9.30570495075558\\
5	8.28729008032914\\
6	6.78554746592443\\
7	4.35381183265742\\
};
\addlegendentry{$T=40$}

\addplot [color=mycolor6, mark=+]
  table[row sep=crcr]{%
1	6.26804388229854\\
2	10.1249027184062\\
3	8.75309744669561\\
4	8.14578362359868\\
5	7.09126033208113\\
6	5.8579040583938\\
7	3.76909992743912\\
};
\addlegendentry{$T=80$}
\end{axis}

\begin{axis}[%
width=5.833in,
height=4.375in,
at={(0in,0in)},
scale only axis,
xmin=0,
xmax=1,
ymin=0,
ymax=1,
axis line style={draw=none},
ticks=none,
axis x line*=bottom,
axis y line*=left,
legend style={legend cell align=left, align=left, draw=white!15!black}
]
\end{axis}
\end{tikzpicture}%